%% file: pp_arxiv.tex
\documentclass[12pt,reqno]{article}
\usepackage[usenames,dvipsnames]{xcolor}
\usepackage{amsmath,mathtools,tikz,pgfplots,soul}

\input{working_arxiv}

\usepackage[T1]{fontenc}
\usepackage{newtxtext}
\usepackage{lmodern,newtxtext,newtxmath}
\usepackage[cal=cm]{mathalfa}

\usepackage{graphicx} % For \raisebox
\usepackage{amsmath}  % For \text
\usepackage{scalerel} % For better sizing

\DeclareRobustCommand{\textcircled}[1]{%
  \scalerel*{%
    \raisebox{-8pt}{\textcircledold{\normalsize #1}}%
  }{X}%
}

\newcommand{\textcircledold}[1]{%
  \ooalign{%
    \hfil#1\hfil\crcr
    \mathhexbox20D%
  }%
}

\usepackage{booktabs,array,tablefootnote,subcaption} 
\usepackage[titletoc]{appendix}
\newcolumntype{C}[1]{>{\centering\arraybackslash}m{#1}}

\makeatletter
\newcommand{\vast}{\bBigg@{4}}
\newcommand{\Vast}{\bBigg@{5}}
\makeatother

\let\Pr\relax
\DeclareMathOperator{\Pr}{Pr}

\DeclareMathOperator{\MPC}{MPC}

\newcommand{\bm}{\mathbf{m}}

\renewcommand{\bs}{\mathbf{s}}

\newcommand\calo{
  \mathchoice
    {\scalebox{.85}{$\scriptstyle\mathcal{O}$}}% \displaystyle
    {\scalebox{.85}{$\scriptstyle\mathcal{O}$}}% \textstyle
    {{\scriptscriptstyle\mathcal{O}}}% \scriptstyle
    {\scalebox{.7}{$\scriptscriptstyle\mathcal{O}$}}%\scriptscriptstyle
  }

\definecolor{azure(colorwheel)}{rgb}{0.0, 0.5, 1.0}
\definecolor{brandeisblue}{rgb}{0.0, 0.44, 1.0}
\definecolor{ceruleanblue}{rgb}{0.16, 0.32, 0.75}
\definecolor{airforceblue}{rgb}{0.36, 0.54, 0.66}
\definecolor{bleudefrance}{rgb}{0.19, 0.55, 0.91}
\definecolor{darkspringgreen}{rgb}{0.09, 0.45, 0.27}
\definecolor{rulecolor}{rgb}{0.13, 0.35, 0.5}

\usepgfplotslibrary{fillbetween,decorations.softclip}
\usetikzlibrary{arrows.meta}
\usetikzlibrary{patterns,calc}

\newcommand{\papertitle}{{\bf Pareto-Improving Pricing: Why 3 Is Better Than 2}\footnote[$*$]{\thanktext}
}

\newcommand{\name}{Zi Yang Kang\footnote[$\dag$]{\affiliationi}\quad \textcircled{r} \quad Piotr Dworczak\footnote[$\ddag$]{\affiliationii}}

\newcommand{\affiliationi}{Department of Economics, University of Toronto; \href{mailto:zy.kang@utoronto.ca}{\tt zy.kang@utoronto.ca}.}
\newcommand{\affiliationii}{University of Cologne; Northwestern University; and Group for Research in Applied Economics (GRAPE); \href{mailto:dworczak.piotr@gmail.com}{\tt dworczak.piotr@gmail.com}.}

\newcommand{\thanktext}{We thank Simon Loertscher, Axel Ockenfels, Paula Onuchic, Mike Ostrovsky, Agathe Pernoud, Bruno Strulovici, Filip Tokarski, Shosh Vasserman, Frank Yang, and several seminar audiences for helpful discussions. We gratefully acknowledge the support received under the ERC Starting Grant IMD-101040122. This research was conducted in part at the Adenauer School of Government (ASG) at the University of Cologne with generous support from the Alfred Landecker Foundation.
}
\newcommand{\abstracttext}{
\noindent We study the design of priority pricing systems with heterogeneous agents in environments in which improving quality for some agents reduces the average quality that can be provided. Contrary to the equity--efficiency tradeoff emphasized in public debates, we show that under economically natural conditions priority pricing can Pareto-improve on an equal-allocation benchmark. Three priority tiers suffice for such an improvement, combining higher quality for a fee, lower quality with compensation, and an intermediate tier at the benchmark quality; two tiers are never enough. Our results provide a framework for overcoming equity--efficiency tensions in applications such as lane pricing, waiting-line design, public provision, and insurance.
\\[2ex]
{\em JEL classification}: D82, D47, D63\\[1ex]
{\em Keywords}: Pareto improvements, mechanism design, inequality-aware market design}
\newcommand{\inserttitle}{\begin{center}\hfill\\[0ex]
\LARGE\papertitle\\[3ex]
\large
\name\\[3ex]
\monthyeardate\today
\\[4ex]
\end{center}}	

\begin{document}
\onehalfspacing
\inserttitle
\begin{abstract}
\abstracttext
\end{abstract}
\setstretch{1.5}
\thispagestyle{empty}

\clearpage

\setcounter{page}{1}

\section{Introduction}
\label{sec:introduction}

Many scarce resources---such as space on highways, bandwidth in communication networks, and access to public services---are allocated through waiting rather than explicit prices.  Since at least \citet{pigou20}, economists have emphasized the inefficiency of such time-based allocation: waiting is socially costly, and many individuals would pay to avoid it, creating scope for gains from trade.  Yet pricing is often resisted on equity grounds because it is seen as replacing time-based rationing with income-based rationing.   

Public debates surrounding ``pay-to-skip-the-line'' systems vividly reflect this concern.  On roads, critics deride priority toll lanes as ``Lexus lanes'' that allow the rich to buy their way out of congestion while others bear the cost of greater delays \citep{anderson05}.  At airports, paid access to expedited security lines---such as CLEAR---has been portrayed as allowing some passengers to cut in line \citep{mull23} and as splitting ``travelers into haves and have-nots'' \citep{stewart22}.  Similar concerns arise when governments offer fast-track treatment for payment.  Critics of investor-citizenship programs, for example, describe them as allowing wealth to purchase accelerated access to a public status for which ordinary applicants must wait \citep{shachar18}.

Our main insight in this paper is that this equity--efficiency tradeoff need not be as stark as these debates suggest. We define the \textit{equity benchmark} as an allocation in which all individuals receive the same quality (e.g., face the same wait time). We assume that departures from this benchmark are costly: giving some individuals higher quality requires both giving others lower quality and lowering quality on average. Even so, prices and priority can be introduced in a way that \emph{Pareto-improves} on the equity benchmark, but only by moving beyond the two-tier systems that dominate practice. With only two tiers---a paid ``fast lane'' and a free regular option---no Pareto improvement is possible: even if the revenue from selling priority is redistributed, some individuals must be worse off than under the equity benchmark. With three tiers, by contrast, every individual can be made strictly better off. Our construction combines a paid high-priority tier, a compensated low-priority tier, and an intermediate tier that preserves the benchmark allocation.  Together, these results help explain why simple priority systems generate backlash: without a third tier, unequal priority necessarily harms some individuals.  At the same time, our framework identifies designs under which everyone benefits from the introduction of prices and priority.

Our formal model is stylized but flexible enough to cover a wide range of applications.  We consider a designer who allocates potentially different qualities of a good to a unit mass of agents with different marginal rates of substitution between money and quality.  In lane pricing, for example, quality decreases with travel time, and agents vary in their distaste for waiting.  Under the equity benchmark, every agent receives the same quality and makes no monetary transfer.  The designer seeks to replace this benchmark with a pricing mechanism---a menu of qualities and monetary transfers---that introduces dispersion in quality.  Application-specific feasibility constraints determine which quality distributions can be implemented.  Rather than model these constraints directly, we impose three assumptions on the feasible set of quality distributions.

First, we assume there is a tradeoff---what we call \emph{friction}---between dispersion and average quality.  Specifically, the degenerate quality distribution associated with the equity benchmark uniquely maximizes average quality.  This assumption captures the idea that reallocating quality to create dispersion is costly.  For example, if travel time is convex in lane traffic, unequal lane loads create dispersion in travel times while raising average travel time. In other applications, friction may arise from curvature in preferences, or from heterogeneity in the costs of serving different types, as under adverse selection.

Second, we assume that this mean--dispersion tradeoff is locally mild.  More precisely, the loss in average quality from creating dispersion can be made arbitrarily small relative to the aggregate quality gain received by agents allocated qualities above the equity benchmark.  Small departures from equal lane loads, for instance, generate first-order dispersion but only a higher-order increase in average travel time.

Third, we assume that the designer can randomize assignments.  Formally, if a quality distribution is feasible, then every mean-preserving contraction of that distribution is also feasible.\footnote{In our applications, we also establish when the Pareto improvement can be implemented without randomizing assignments.}

Under these assumptions, we prove two main results. On the one hand, no mechanism with only two tiers---that is, two quality levels---can Pareto-improve on the equity benchmark.  On the other hand, there exists a three-tier mechanism that strictly Pareto-improves on the equity benchmark. Thus, three tiers are the minimum required to make everyone better off---a modest but practically feasible level of complexity.

For intuition, note that priority pricing creates \textit{sorting gains} by reallocating quality toward agents with higher willingness to pay, but the resulting dispersion lowers average quality.  A Pareto improvement thus requires extracting enough of these gains through prices and redistributing the revenue to compensate every agent for the loss in average quality.

With only two tiers, revenue from the high tier finances a subsidy to the low tier.  The threshold type is indifferent between the two tiers, so the price premium for the higher tier must equal this type's willingness to pay for the quality upgrade.  Budget balance then implies that the subsidy exactly compensates the threshold type for receiving the lower quality rather than the mechanism's \emph{average} quality.  Consequently, the threshold type obtains precisely the value of average quality.  Since average quality is lower than under the equity benchmark, the threshold type is strictly worse off.

Adding a third tier relaxes this tension by creating two distinct pricing margins.  The two quality upgrades are priced at the willingness to pay of distinct threshold types, so quality improvements at the top are valued at a higher rate than the quality reductions at the bottom. This creates additional sorting gains that can be extracted and redistributed.  As long as the threshold types remain bounded away from each other, these gains are first order in the size of the quality dispersion, while the corresponding loss in average quality is of higher order.  The gains can therefore compensate every agent for the loss in average quality. 

Our applications illustrate the breadth and flexibility of the framework. Our assumptions, and hence the main results, apply across environments with different allocation technologies and sources of friction. Beyond our running example of lane pricing, we study the design of wait lines---for instance, for airport screening or access to public services---where our results explain how a design with three differentiated lines can avoid the equity--efficiency tradeoff, whereas a two-tier system cannot. In the problem of allocating goods with heterogeneous quality, such as public housing, we show that offering at least three types of units---for example, differentiated by size---can generate a Pareto improvement over uniform provision. Finally, we consider insurance design under adverse selection and show how our framework and main result extend to environments in which the resource cost of quality depends on the recipient.

The remainder of this paper is organized as follows. We conclude this section by reviewing the related literature. \Cref{sec:framework} introduces the formal model, and \Cref{sec:main} presents the main results.  We apply and generalize these results in \Cref{sec:applications}, and \Cref{sec:conclusion} concludes.

\subsection{Related Literature}

The pricing of priority in access to goods and services---and its distributional consequences---has received sustained attention in economics. As we discuss below, several papers show that priority pricing can Pareto-improve on equal access or laissez-faire in particular environments.

Our contribution to this literature is to propose a general framework for priority pricing and to characterize the minimum number of tiers required for a Pareto improvement.  We model the supply side through the set of feasible quality distributions, which isolates two channels through which priority pricing affects welfare: \textit{(i)} changes in average quality and \textit{(ii)} sorting gains from allocating quality to agents with higher willingness to pay.  We focus on environments with a mean--dispersion tradeoff, in which creating the dispersion needed for sorting \textit{lowers} average quality, so any Pareto improvement must be driven by sorting gains.  To the best of our knowledge, our paper is the first to show that three, rather than two, priority tiers are necessary to achieve a Pareto improvement in such environments.

\citet{chaowilson87} and \citet{gershkovschweinzer10} show that selling priority and redistributing the resulting revenue can generate Pareto improvements.  \citeauthor{gershkovschweinzer10} study a model in which agents differ in marginal utility of time and hold property rights to potentially random positions in a queue.  They characterize when these positions can be efficiently reordered through an incentive-compatible, budget-balanced mechanism while leaving every agent weakly better off in expectation.  Most relevant for us, they show that a random queue order---which corresponds to our equity benchmark---can be Pareto-improved upon.\footnote{\citet{gershkovwinter23} make the complementary point that priority service can uniformly harm consumers when access is controlled by a profit-maximizing monopolist.}  With total waiting time fixed and linear disutility from waiting, their queue-reordering model permits sorting gains without reducing average quality. Thus, it does not feature the mean--dispersion tradeoff that is at the heart of our results.

A sizable literature studies the distributional consequences of congestion pricing in the bottleneck model of \citet{vickrey69}, including \citet{arnottetal94}, \citet{vandenbergverhoef11}, \citet{hall18}, and \citet{bobbioetal21}.  These models likewise do not generally feature the mean--dispersion tradeoff captured by our friction assumption.  In fact, because agents can adjust their departure times in response to prices in the \citeauthor{vickrey69} framework, congestion pricing can reduce average congestion---or, in the language of our paper, \emph{increase} average quality.  This channel is often central to obtaining a Pareto improvement and explains why some of these papers obtain Pareto improvements even before redistributing revenue.  \citet{hall18}, for example, shows that preventing the fall in road throughput associated with hypercongestion can yield a Pareto improvement even with identical agents and before revenue is redistributed. Our friction assumption, by contrast, requires dispersion to reduce average quality relative to the equity benchmark. Pareto improvements must therefore be driven by sorting gains, with pricing revenue redistributed so as to compensate every agent for the resulting loss in average quality.\footnote{Because our positive result predicts a \textit{strict} Pareto improvement, it remains valid when transfers are imperfect and a sufficiently small share of revenue is dissipated.}

More broadly, our paper contributes to a growing market-design literature on congestion pricing, including work by \citet{ostrovskyschwarz18}, \citet{cramtonetal19}, \citet{beheshtianetal20}, and \citet{ostrovskyyang24}.\footnote{An active empirical literature also quantifies the welfare effects of congestion pricing; see, for example,  \citet{hall21}, \citet{kreindler24}, \citet{cookli25}, \citet{cooketal26}, \citet{ateretal26}.}

Finally, our paper relates to the recent literature on \emph{inequality-aware market design}.  Building on classical work by \citet{weitzman77}, \citet{spence77}, and \citet{nicholszeckhauser82}, this literature studies mechanism design under redistributive objectives; contributions include work by \citet{condorelli13}, \citet{dworczaketal21}, \citet{kang23}, and \citet{akbarpouretal24}, among many others.  The paper most closely related to ours within this literature is \citet{kangwatt26}.  They study public provision when agents also have access to a private market.  The private option imposes a constraint mathematically similar to Pareto improvement.  While this literature characterizes optimal mechanisms for a specified social welfare function, we instead identify conditions under which a Pareto improvement exists without specifying interpersonal welfare comparisons.

\newpage
\section{Framework}
\label{sec:framework}

We develop an abstract framework for the design of Pareto-improving pricing systems.  The framework is intentionally parsimonious: it is broad enough to encompass a wide range of applications, yet sufficiently structured to isolate the key economic forces that drive our results.

\subsection{Setup}

There is a unit mass of agents, each of whom demands one unit of a good.  Units of the good may differ in their physical quality $q\in\R_+$.  Each agent privately observes his type $r$.  The distribution of types in the population is $G$, which has positive continuous density $g$ on a compact interval $[\und r,\BAR r]\subset\R_+$.  An agent of type $r$ who is assigned a unit of physical quality  $q$ and makes a payment $p\in\R$ obtains utility
\[r v(q)-p,\]
where a negative payment is interpreted as a rebate to the agent.  The function $v:\R_+\to [0,1]$ is assumed to be continuously differentiable, increasing,\footnote{Throughout, we use ``increasing'' and ``decreasing'' to mean strictly increasing and strictly decreasing, respectively; weak monotonicity is denoted by ``nondecreasing'' and ``nonincreasing.''} and concave, with $v(0)=0$. 

We refer to $Q:=v(q)$ as the good's \emph{effective quality} (or simply \emph{quality}, when there is no ambiguity).  Agent utility is $rQ-p$, so the agent's type is his marginal rate of substitution between effective quality and money.  The distinction between physical and effective quality separates the objective characteristics of the good from their utility consequences for agents.  In particular, agents are risk-neutral over lotteries in effective quality, but may be risk-averse over lotteries in physical quality (when $v$ is strictly concave).

We use the following running example throughout the paper to illustrate the abstract framework. 

\begin{example*}
Agents are commuters traveling from a common origin to a common destination.  A trip entails time $t\in[\und t,\BAR t]$ spent in traffic, where $0\leq\und t<\BAR t<1$.  Normalizing the physical quality of a trip with no time spent in traffic to $1$, define the physical quality of a trip with travel time $t$ as $q=1-t$.  Thus, shorter travel time corresponds to higher physical quality.  When $v$ is strictly concave, agents are risk-averse over travel time: they strictly prefer a certain travel time to any nondegenerate lottery with the same mean.
\end{example*}

\subsection{Outcomes}

A designer assigns each agent a possibly degenerate lottery over physical qualities.  Because utility is linear in effective quality $Q$, an agent's payoff depends on his assigned lottery only through its expected quality.  We can therefore summarize the physical allocation by a distribution $F\in\Delta([0,1])$ of expected qualities available for assignment, where $1-F(Q)$ is the mass of agents who can be assigned expected quality strictly greater than $Q$.\footnote{This representation remains without loss of generality when goods are scarce: an agent who receives no good can be treated as receiving physical quality zero.}

Different applications restrict the feasible distributions of expected qualities in different ways.  We represent these restrictions by a set $\calF\subseteq\Delta([0,1])$ from which the designer must select $F$.\footnote{For simplicity, we abstract from variable production costs.  \Cref{sec:insurance} shows how our results accommodate such costs.}  The set $\calF$ is the central object of our framework; \Cref{sec:feasible} introduces assumptions on $\calF$ that generate the tradeoff between dispersion and average quality underlying our results. 

Each feasible distribution $F\in\calF$ uniquely determines, up to a $G$-null set, a budget-balanced and incentive-compatible deterministic reduced-form mechanism in expected qualities and payments.\footnote{Here, ``deterministic'' refers to the reduced-form pair $(Q,p)$ assigned to each type and does not rule out randomization over physical qualities, which can be incorporated into $\calF$.  This restriction is without loss because utility is linear in expected quality and payments.}  Incentive compatibility requires expected quality to be nondecreasing in type \citep{myerson81}.  Hence, the allocation rule is
\[Q(r)=F^{-1}(G(r)),\]
the unique nondecreasing allocation rule whose distribution is $F$, up to $G$-null sets.  Given this allocation rule, the envelope formula determines payments up to an additive constant \citep{milgromsegal02}:
\[p(r)=rQ(r) - U - \int_{\und r}^r Q(s)\,\dd s,\]
where $U$ is uniquely chosen to satisfy budget balance:
$\int_{\und r}^{\BAR r}p(r)\,\dd G(r)=0.$
We say that the resulting mechanism $\(Q(\cdot),p(\cdot)\)$ \emph{implements} $F$.

\subsection{Pareto Improvements}

We define the \emph{equity benchmark} as the degenerate quality distribution $\d_{Q_0}$, where
\[Q_0\coloneq\sup_{F\in\calF}\E_F[Q].\]
We assume that $\d_{Q_0}\in \calF$.\footnote{This assumption is without loss of generality if $\calF$ is compact and the designer can randomize.}
Under the equity benchmark, average quality is maximized while every agent receives expected quality $Q_0$ and pays zero; the resulting direct mechanism is budget-balanced and incentive-compatible. As we show in \Cref{sec:applications}, this benchmark often coincides with the laissez-faire non-price allocation---that is, the equilibrium allocation in a system without monetary transfers.

The designer seeks to construct a Pareto improvement over the equity benchmark.  For any distribution $F\in\calF$, we say that $F$ is a \emph{Pareto improvement over the equity benchmark} (or simply a \emph{Pareto improvement}) if its implementing mechanism $\(Q(\cdot),p(\cdot)\)$ satisfies
\begin{equation}\label{eq:PI}
rQ(r)-p(r)\geq rQ_0,
\qquad\text{for every }r\in[\und r,\BAR r],\tag{PI}
\end{equation}
with strict inequality for a positive mass of agents.  We say that $F$ is a \emph{strict Pareto improvement} if every agent strictly prefers its implementing mechanism to the equity benchmark---that is, if the inequality in \eqref{eq:PI} is strict for every type $r\in[\und r,\BAR r]$.

Our focus on Pareto improvements reflects an intentionally conservative welfare criterion.  We make no interpersonal comparisons of utility and posit no social welfare function.  Unlike much of the existing literature, which studies optimal allocations under explicit welfare objectives, we ask only whether prices can realize Pareto gains, subject to budget balance and incentive compatibility, without invoking distributional judgments.

\subsection{Feasible Policies}\label{sec:feasible}

We complete the description of our framework by imposing three economically motivated restrictions on the set of feasible expected-quality distributions $\calF\subseteq\Delta([0,1])$. 

We first motivate these restrictions using our running example.

\begin{example*}
There are $N\geq 2$ lanes on a highway.  If a mass $m_i$ of agents is assigned to lane $i$, each agent in that lane experiences travel time $w(m_i)$, where $w:[0,1]\to[\und t,\BAR t]$ is continuously differentiable, increasing, and convex.  These assumptions capture congestion effects: travel time increases with lane usage at a nondecreasing rate.

Given a vector of lane loads $\bm=(m_1,\ldots,m_N)\in\Delta^{N-1}$, the distribution of travel times is $\sum_{i=1}^N m_i\d_{w(m_i)}$.  Because travel time $t$ corresponds to effective quality $v(1-t)$, the induced distribution of effective qualities is
\[F^{\bm} \coloneq \sum_{i=1}^N m_i\d_{v(1-w(m_i))}.\]

Starting from any lane-load vector $\bm$, the designer can randomize agents' assignments across lanes while preserving the marginal distribution of realized physical qualities.  The distributions of expected qualities that can be generated  this way are precisely the mean-preserving contractions of $F^{\bm}$, which we denote by $\MPC(F^{\bm})$.\footnote{Strassen's theorem \citep[Theorem 3.4.2]{mullerstoyan02} implies that there exist random variables $X$ and $\bar X$ such that $X\sim F$, $\bar X\sim\bar F$, and $\E[X\mid\bar X]=\bar X$ if and only if $\bar F$ is a mean-preserving contraction of $F$.  Here, $X$ is realized effective quality and $\bar X$ is expected effective quality.  For any $F\in\Delta([0,1])$, define
\[\MPC(F)\coloneq\left\{\bar F\in\Delta([0,1]):\int_0^x\bar F^{-1}(u)\,\dd u\geq\int_0^x F^{-1}(u)\,\dd u\text{ for every }x\in[0,1],\text{ with equality at }x=1\right\}.\]}
Consequently, the feasible set of expected quality distributions is
\[\calF = \left\{F\in\Delta([0,1]):F\in\MPC(F^{\bm})\text{ for some }\bm\in\Delta^{N-1}\right\}.\]
\end{example*}

Rather than committing to a particular technology, we impose three axioms on the feasible set $\calF$.

\begin{axiom}[Randomization]\label{ax:randomization}
The feasible set $\calF$ is closed under mean-preserving contractions: if $F\in\calF$, then $\MPC(F)\subseteq\calF$.
\end{axiom}

\Cref{ax:randomization} formalizes the designer's ability to randomize assignments.  Starting from any feasible distribution, the designer can reduce the dispersion in agents' expected qualities without changing their average.  The axiom holds by construction in the running example.

\begin{axiom}[Friction]\label{ax:friction}
The equity benchmark $\d_{Q_0}$ uniquely maximizes average quality: for every $F\in \calF$,  if $F\ne\d_{Q_0}$, then  $\E_F[Q]<Q_0.$  
\end{axiom}

\Cref{ax:friction} captures a friction in the allocation technology: any feasible departure from the equity benchmark lowers average quality. The axiom therefore imposes a local tradeoff between dispersion and average quality around the benchmark, thereby making improvements over the equity benchmark more difficult to achieve.

The axiom holds in the running example under the maintained assumptions.  To see this, define the function $\phi(m)\coloneq mv(1-w(m))$, and note that the average quality generated by a lane-load vector $\bm$ is $\sum_{i=1}^N\phi(m_i)$.
Since $w$ is increasing and convex and $v$ is increasing and concave, $\phi$ is strictly concave.  By Jensen's inequality, for any $\bm\in \Delta^{N-1}$,
\[\sum_{i=1}^N\phi(m_i)\leq N\phi\!\(\frac1N\)=v\!\(1-w\!\(\frac1N\)\),\]
with equality if and only if $m_i=1/N$ for every $i\in\{1,\ldots,N\}$.  Hence, $Q_0=v(1-w(1/N))$ satisfies the requirements of the axiom, and the equal assignment of agents to lanes induces the unique feasible distribution with average quality $Q_0$. Intuitively, because travel times increase with lane usage at a nondecreasing rate and agents are weakly risk-averse over travel time, equalizing lane usage uniquely maximizes average effective quality.\footnote{Note that neither strict risk aversion nor strictly convex congestion effects are required for \Cref{ax:friction} to hold.}

\begin{axiom}[Smoothness]\label{ax:smoothness}
The tradeoff between dispersion and average quality around the equity benchmark $\d_{Q_0}$ is locally mild: for every $\e>0$, there exists $F\in \calF$ such that 
\begin{equation}\label{eq:smoothness}
\E_F[\(Q-Q_0\)_+]>0\AND \frac{Q_0-\E_{F}[Q]}{\E_F[\(Q-Q_0\)_+]} \leq \e.\tag{S}
\end{equation}
\end{axiom}

\Cref{ax:smoothness} provides a counterpoint to \Cref{ax:friction}.  Whereas friction requires every feasible departure from the equity benchmark to lower average quality, smoothness requires that this loss, $Q_0-\E_{F}[Q]$, can be made arbitrarily small relative to the cumulative improvement $\E_F[\(Q-Q_0\)_+]$ above $Q_0$. If the uniqueness requirement in \Cref{ax:friction} were violated, so that a nondegenerate feasible distribution had average quality $Q_0$, condition~\eqref{eq:smoothness} would hold trivially. 

\Cref{ax:smoothness} also holds in the running example.  Starting from equal lane loads, transfer a mass $\eta$ of agents from one lane to another.  The first-order effects on average quality cancel because $w$ and $v$ are continuously differentiable.  The resulting loss in average quality is therefore $\calo(\eta)$, whereas the cumulative improvement over $Q_0$ is of order~$\eta$ because effective quality is locally strictly decreasing in lane usage.  Hence, the ratio in condition~\eqref{eq:smoothness} converges to zero as $\eta\downarrow0$.  \Cref{sec:lane} and \Cref{app:lanes} provide a formal proof.

\section{Main Results}
\label{sec:main}

Our main results characterize the minimum number of tiers required to construct a Pareto improvement. To formalize simplicity in a policy-relevant way, we count the number of distinct expected-quality levels offered by the mechanism.

\begin{definition}\label{def:tiers}
A distribution $F\in\calF$ offers $n$ tiers if $\abs{\supp(F)}=n$.
\end{definition}

Any Pareto improvement must offer at least two tiers: a one-tier mechanism assigns the same quality $Q\leq Q_0$ to every agent, while incentive compatibility and budget balance require payments to be zero. 

\subsection{Statement and Discussion of Main Results}

Our first result shows that two tiers do not suffice.

\begin{theorem}\label{thm:two}
Under \Cref{ax:friction}, there does not exist a Pareto improvement with two tiers.
\end{theorem}

Our second result shows that three tiers suffice.

\begin{theorem}\label{thm:three}
Under \Cref{ax:randomization,ax:smoothness}, there exists a strict Pareto improvement with three tiers.
\end{theorem}

Together, \Cref{thm:two,thm:three} establish that, under our three axioms, three is the minimum number of tiers required for a Pareto improvement.

These results are particularly relevant because two-tier priority systems are widespread in practice. Examples include express lanes on roads and fast-track options in airport security, visa processing, and health care. Such systems are appealing for their simplicity and are often defended as a way to generate allocative gains by allowing agents who value priority more highly to pay for it.

In this light, \Cref{thm:two} is a negative result: any nontrivial two-tier priority system must make some agents strictly worse off than under the equity benchmark.  Relative to giving every agent the mechanism's average quality at no charge, sorting by willingness to pay benefits everyone except the agent indifferent between the two tiers.  This agent receives no sorting gain to offset the loss in average quality and is therefore worse off.  Even without \Cref{ax:friction}, a two-tier system cannot make every agent \textit{strictly} better off.

The negative result of \Cref{thm:two} therefore helps explain why many two-tier priority systems have generated public backlash.  When congestion or other frictions create a tradeoff between dispersion and average quality, allowing some users to purchase priority necessarily harms others, even if the resulting revenue is redistributed.  Moreover, if type measures the marginal value of reduced waiting time relative to the marginal value of money (e.g., \citealp{dworczaketal21}), the intermediate types harmed by a two-tier system may include low- and middle-income users who place a high value on time but also have a high marginal value of money.

\Cref{thm:three} provides the positive counterpart: three tiers are sufficient to make every agent \textit{strictly} better off.  We construct a mechanism with an intermediate tier that provides expected quality exactly equal to $Q_0$.  This tier separates the threshold between the low and intermediate tiers from the threshold between the intermediate and high tiers.  The quality jump at each threshold generates sorting gains for the agent indifferent at the other, so that together they generate sorting gains for every agent.  \Cref{ax:randomization,ax:smoothness} allow us to construct such a mechanism with gains that outweigh the loss in average quality.

This positive result shows that achieving a Pareto improvement does not require a complex or finely targeted pricing system: under our three axioms, three appropriately designed priority tiers suffice. Policy debates framed as a binary choice between uniform service and a two-tier priority system are therefore unnecessarily restrictive.

\begin{example*}
Because our running example satisfies \Cref{ax:randomization,ax:friction,ax:smoothness},  Theorems~\ref{thm:two} and~\ref{thm:three} apply. A conventional two-tier lane-pricing system cannot Pareto-improve on equal access: creating one faster lane necessarily leaves some commuters strictly worse off, even if toll revenues are redistributed. By contrast, three tiers of expected quality can make every commuter strictly better off. \Cref{sec:lane} returns to this application and studies its implementation in more detail.
\end{example*}

\subsection{Proofs of Main Results}\label{sec:proofs}

We begin by decomposing each agent's utility gain from priority pricing into two components: the effect of the change in average quality and the gains from sorting.  For any feasible distribution $F$ and its implementing mechanism $(Q,p)$, incentive compatibility and budget balance imply that
\begin{align}
rQ(r)-p(r)-rQ_0 
&= \underbrace{r\left[\E_F[Q]-Q_0\right]}_{\text{effect of change in average quality}} \notag\\
&\qquad+ \underbrace{\int_{\und r}^{r}\(r-s\)G(s)\,\dd Q(s) + \int_r^{\BAR r}\(s-r\)\left[1-G(s)\right]\,\dd Q(s)}_{\text{gains from sorting}}.\label{eq:decomposition}
\end{align}
Here, $\dd Q$ denotes the Stieltjes measure induced by the nondecreasing allocation rule $Q$.

The decomposition shows that each interior quality jump generates sorting gains for every type except the threshold type itself.  The two integrals in \cref{eq:decomposition} measure these gains relative to giving all agents the mechanism's average quality at no charge.  Jumps at thresholds below $r$ contribute positively to the first integral, and jumps at thresholds above $r$ contribute positively to the second.

With only two tiers, the single quality jump leaves its threshold type with no sorting gain.  In \Cref{fig:two}, this is the type $r_0$ at which the two linear segments of the utility schedule meet.  By \cref{eq:decomposition}, his utility is therefore $r_0\E_F[Q]$.  \Cref{ax:friction} implies that $\E_F[Q]<Q_0$, so this point lies strictly below the utility schedule under the equity benchmark.

\begin{figure}[t!]
\centering
\begin{tikzpicture}[scale=1]

\draw[line width=1.pt,densely dashed,color=black!50]
    (6,.15) -- (6,3);

\draw[scale=1,domain=1:10.5,densely dashed,variable=\x,
      line width=2.pt,color=black!50]
    plot ({\x},{.5*\x});
\draw (2.2,1.1) node[below right]
    {$\color{black!50}rQ_0$};

\draw[-{latex[scale=1.2]},line width=1.5pt]
    (0,0) -- (11.5,0) node[below] {$r$};
\draw[-{latex[scale=1.2]},line width=1.5pt]
    (0,-.5) -- (0,6.8) node[above] {utility};

\draw[line width=5pt,magenta]
    (1,0) -- (3.5,0)
    node[below] {\small $\color{magenta}Q_L$\vphantom{b}}
    -- (6,0);
\draw[line width=5pt,darkspringgreen]
    (6,0) -- (8.25,0)
    node[below] {\small $\color{darkspringgreen}Q_H$\vphantom{b}}
    -- (10.5,0);

\draw[scale=1,domain=1:6,variable=\x,
      line width=2.pt,color=ceruleanblue]
    plot ({\x},{1.5+.2*\x});
\draw[scale=1,domain=6:10.5,variable=\x,
      line width=2.pt,color=ceruleanblue]
    plot ({\x},{-1.5+.7*\x});

\draw (2.4,2.2) node[above]
    {\small $\color{ceruleanblue}rQ_L-p_L$};
\draw (8.4,5.0) node[above]
    {\small $\color{ceruleanblue}rQ_H-p_H$};

\draw[line width=1.5pt] (1,.15) -- (1,-.15)
    node[below] {$\und r$};
\draw[line width=1.5pt] (6,.15) -- (6,-.15)
    node[below] {$r_0$};
\draw[line width=1.5pt] (10.5,.15) -- (10.5,-.15)
    node[below] {$\BAR r$};

\end{tikzpicture}
\caption{Agent utility in a two-tier priority system.  The type $r_0$ is strictly worse off than under the equity benchmark.}
\label{fig:two}
\end{figure}
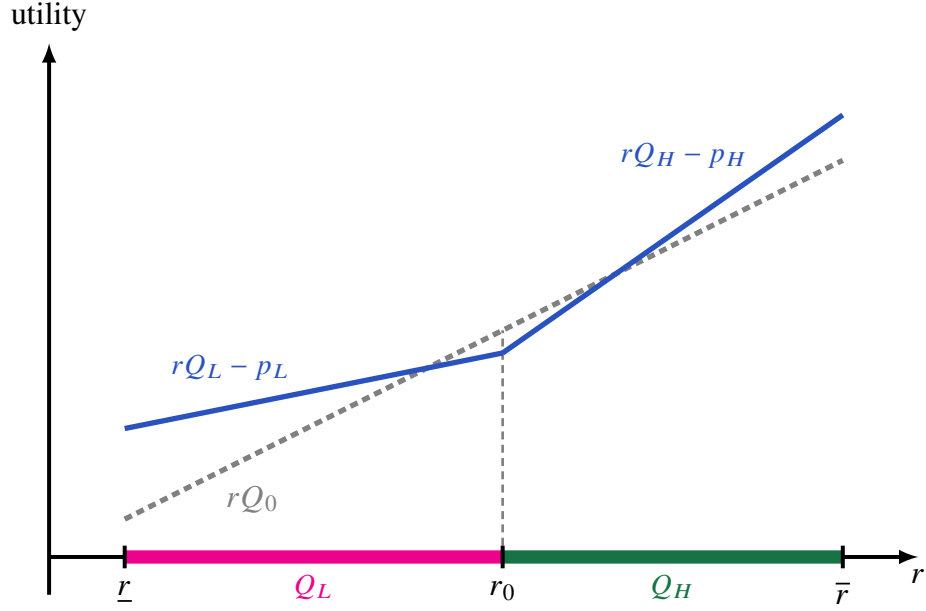

\paragraph{Proof of \texorpdfstring{\Cref{thm:two}}{Theorem 1}.}
We begin by proving \cref{eq:decomposition} for any feasible distribution $F$ and its implementing mechanism $(Q,p)$.  By the envelope theorem \citep{milgromsegal02}, incentive compatibility implies that, for every $r,s\in[\und r,\BAR r]$,
\[rQ(r) - p(r) = sQ(s) - p(s) + \int_{s}^r Q(t)\,\dd t.\]
Averaging over $s$ and using budget balance gives
\begin{align*}
rQ(r)-p(r)
&=\int_{\und r}^{\BAR r}sQ(s)\,\dd G(s)
+\int_{\und r}^{\BAR r}\left[\int_s^r Q(t)\,\dd t\right]\dd G(s)\\
&=\int_{\und r}^{\BAR r}sQ(s)\,\dd G(s) + \int_{\und r}^{r}G(s)Q(s)\,\dd s - \int_r^{\BAR r}[1-G(s)]Q(s)\,\dd s.
\end{align*}
Stieltjes integration by parts then yields the two sorting terms in \cref{eq:decomposition}.

Every feasible two-tier mechanism has an interior threshold $r_0>0$, so our argument above using \cref{eq:decomposition} applies.  By continuity of the utility schedule, a positive mass of agents near $r_0$ is strictly worse off than under the equity benchmark.\qed

Even without \Cref{ax:friction}, the threshold type cannot be strictly better off, since $\E_F[Q]\le Q_0$ by the definition of $Q_0$.

\begin{corollary}
There does not exist a strict Pareto improvement with two tiers.     
\end{corollary}

\begin{figure}[t!]
\centering
\begin{tikzpicture}[scale=1]

\draw[line width=1.pt,densely dashed,color=black!50]
    (4,.15) -- (4,2.4);
\draw[line width=1.pt,densely dashed,color=black!50]
    (7,.15) -- (7,3.6);

\draw[scale=1,domain=1:10.5,densely dashed,variable=\x,
      line width=2.pt,color=black!50]
    plot ({\x},{.4*\x});
\draw (2.2,.88) node[below right]
    {$\color{black!50}rQ_0$};

\draw[-{latex[scale=1.2]},line width=1.5pt]
    (0,0) -- (11.5,0) node[below] {$r$};
\draw[-{latex[scale=1.2]},line width=1.5pt]
    (0,-.5) -- (0,6.8) node[above] {utility};

\draw[line width=5pt,magenta]
    (1,0) -- (2.5,0)
    node[below] {\small $\color{magenta}Q_L$\vphantom{b}}
    -- (4,0);
\draw[line width=5pt,orange]
    (4,0) -- (5.5,0)
    node[below] {\small $\color{orange}Q_0$\vphantom{b}}
    -- (7,0);
\draw[line width=5pt,darkspringgreen]
    (7,0) -- (8.75,0)
    node[below] {\small $\color{darkspringgreen}Q_H$\vphantom{b}}
    -- (10.5,0);

\draw[scale=1,domain=1:4,variable=\x,
      line width=2.pt,color=ceruleanblue]
    plot ({\x},{1.6+.2*\x});
\draw[scale=1,domain=4:7,variable=\x,
      line width=2.pt,color=ceruleanblue]
    plot ({\x},{.8+.4*\x});
\draw[scale=1,domain=7:10.5,variable=\x,
      line width=2.pt,color=ceruleanblue]
    plot ({\x},{-1.3+.7*\x});
\draw (10.7,6.05) node[right]
    {$\color{ceruleanblue}rQ(r)-p(r)$};

\draw[{latex[scale=.8]}-{latex[scale=.8]},line width=1.pt]
    (5.5,2.2) -- (5.5,3)
    node[midway,right,yshift=0.5ex] {\footnotesize $-p_0$};

\draw[line width=1.5pt] (1,.15) -- (1,-.15)
    node[below] {$\und r$};
\draw[line width=1.5pt] (4,.15) -- (4,-.15)
    node[below] {$r_L$};
\draw[line width=1.5pt] (7,.15) -- (7,-.15)
    node[below] {$r_H$};
\draw[line width=1.5pt] (10.5,.15) -- (10.5,-.15)
    node[below] {$\BAR r$};

\end{tikzpicture}
\caption{Agent utility in a three-tier Pareto improvement.}
\label{fig:three}
\end{figure}
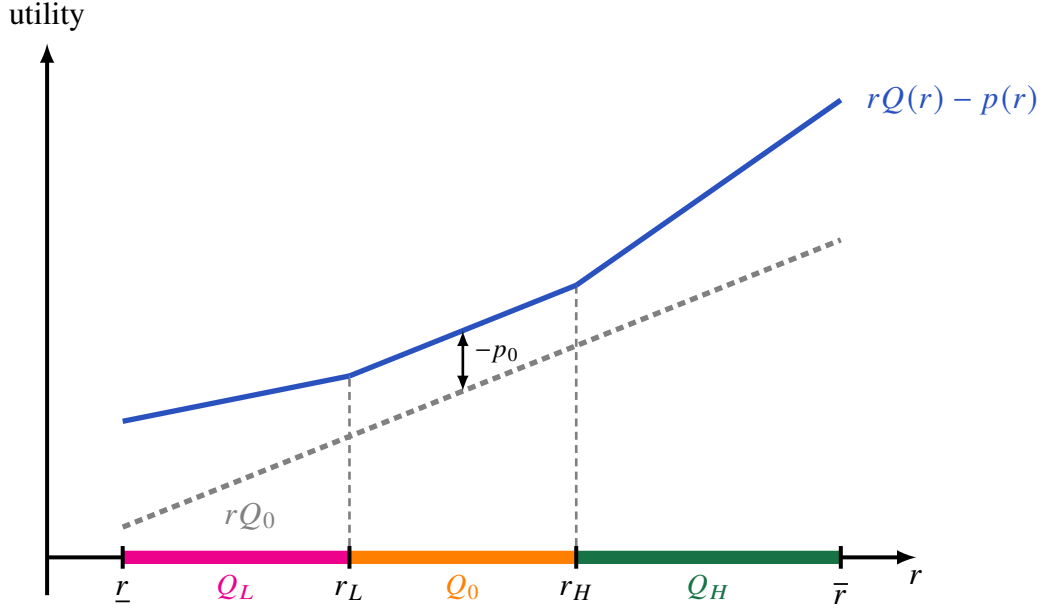

For a three-tier mechanism whose intermediate quality is $Q_0$, a strict Pareto improvement occurs precisely when the intermediate tier is subsidized.  Consider a feasible distribution $F$ with qualities $Q_L<Q_0<Q_H$, as illustrated in \Cref{fig:three}. Let $r_L<r_H$ denote its threshold types and $p_0$ the payment in the intermediate tier. Utility gains relative to the equity benchmark decrease with type below $r_L$, remain constant between $r_L$ and $r_H$, and increase above $r_H$, so they are smallest for agents in the intermediate tier. Their utility gain is $-p_0$, the vertical difference between the two utility schedules in \Cref{fig:three}. Writing $\pi_H=1-G(r_H)$ for the mass assigned to the high tier, \cref{eq:decomposition} evaluated at $r_L$ gives
\begin{equation}\label{eq:middle-tier-gain}
-p_0=\(r_H-r_L\)\pi_H\(Q_H-Q_0\)-r_L\left[Q_0-\E_F[Q]\right].
\end{equation}

The separation between the threshold types allows sorting gains to offset the loss in average quality.  At $r_L$, the jump to the intermediate tier contributes no sorting gain, while the jump to the high tier contributes the first term in \cref{eq:middle-tier-gain}.  Because $r_H>r_L$, the quality gain in the high tier can be priced at a higher marginal willingness to pay than is required to compensate for the quality loss in the low tier.

Our construction keeps the threshold types separated while making the loss in average quality small relative to the sorting gains.  We first use \Cref{ax:smoothness} to select a feasible distribution with a small loss in average quality relative to the cumulative improvement above $Q_0$.  We then use \Cref{ax:randomization} to contract it into three tiers, with intermediate quality $Q_0$.  The contraction preserves average quality and retains enough of the cumulative improvement to finance a subsidy for the intermediate tier.

\paragraph{Proof of \texorpdfstring{\Cref{thm:three}}{Theorem 2}.}
Fix the threshold types $r_L:=G^{-1}(1/3)$ and $r_H:=G^{-1}(2/3)$.  Choose $\e>0$ sufficiently small, such that
\[0<\e<\frac{r_H-r_L}{2r_H+r_L}<\frac12.\]
By \Cref{ax:smoothness}, there exists $F\in\calF$ such that
\[\E_F[\(Q-Q_0\)_+] > 0\AND \frac{Q_0-\E_F[Q]}{\E_F[\(Q-Q_0\)_+]} \leq \e.\]
Define 
\[Q_L\coloneq Q_0-\E_F[\(Q_0-Q\)_+]\AND Q_H\coloneq 2\E_F[Q]+\E_F[\(Q-Q_0\)_+]-Q_0.\]
It is easy to see that $0\leq Q_L<Q_0$; by our earlier choice of $\e$, $Q_0<Q_H$.  The definition of $Q_0$ implies that
\[Q_H = Q_0 + \E_F[\(Q-Q_0\)_+] - 2\left[Q_0-\E_F[Q]\right]\leq Q_0 + \E_F[\(Q-Q_0\)_+]\leq 1.\]
Thus, we can define the three-tier distribution
\[F^{(3)}\coloneq\frac13\d_{Q_L} + \frac13\d_{Q_0} + \frac13\d_{Q_H}.\]

We now show that the distribution $F^{(3)}$ is a mean-preserving contraction of $F$ and hence feasible by \Cref{ax:randomization}.  For $z\leq Q_0$, Jensen's inequality implies that 
\begin{align*}
    \E_{F^{(3)}}[\(z-Q\)_+] = \frac13\(z-Q_L\)_+ 
    &=\frac13\left[z-Q_0 + \E_F[\(Q_0-Q\)_+]\right]_+\\
    &\leq\frac13\E_F[\(z-Q\)_+]\leq \E_F[\(z-Q\)_+].
\end{align*}
Similarly, for $z>Q_0$,
\begin{align*}
    \E_{F^{(3)}}[\(Q-z\)_+] = \frac13\(Q_H-z\)_+ 
    &\leq\frac13\left[Q_0 + \E_F[\(Q-Q_0\)_+]-z\right]_+\\
    &\leq\frac13\E_F[\(Q-z\)_+]\leq \E_F[\(Q-z\)_+].
\end{align*}
Since $F^{(3)}$ and $F$ share the same mean, it follows that $F^{(3)}$ is a mean-preserving contraction of $F$.

Finally, we show that the implementing mechanism strictly improves every agent's utility.  By \cref{eq:middle-tier-gain}, its intermediate-tier payment satisfies
\begin{align*}
    -p_0
    &=\frac{r_H-r_L}{3}\E_F[\(Q-Q_0\)_+]-\frac{2r_H+r_L}{3}\left[Q_0-\E_F[Q]\right]\\
    &\geq\frac{\E_F[\(Q-Q_0\)_+]}{3}\left[r_H-r_L-\(2r_H+r_L\)\e\right]>0.
\end{align*}
Here, the final inequality follows from our earlier choice of $\e$.  We conclude that $F^{(3)}$ results in a strict Pareto improvement, as claimed.\qed

The proof of \Cref{thm:three} shows that a strict Pareto improvement can be achieved with three equally sized tiers.  The tier masses can thus be fixed independently of the distribution of agents' types.

\Cref{eq:middle-tier-gain} also characterizes which feasible three-tier distributions with intermediate quality $Q_0$ yield a strict Pareto improvement.

\begin{corollary}\label{cor:three}
A feasible three-tier distribution
$F=\pi_L\d_{Q_L}+\(1-\pi_L-\pi_H\)\d_{Q_0}+\pi_H\d_{Q_H}$, where $Q_L<Q_0<Q_H$, $\pi_L,\pi_H>0$, and $\pi_L+\pi_H<1$,  yields a strict Pareto improvement over the equity benchmark if and only if 
\[
\frac{Q_H-Q_0}{Q_0-Q_L}
>
\frac{\pi_L}{\pi_H}
\frac{G^{-1}(\pi_L)}{G^{-1}(1-\pi_H)}.
\]
\end{corollary}

\Cref{cor:three} relates the scope for a Pareto improvement to heterogeneity in agents' rates of substitution.  As $\pi_L=\pi_H\downarrow0$, the right-hand side of the condition converges to $\und r/\BAR r$, a simple measure of this heterogeneity and hence of the potential gains from sorting.\footnote{This observation is a point of contact with \citet{ateretal26}, who emphasize that driver heterogeneity is crucial for assessing the welfare gains of congestion pricing.}  In \Cref{sec:insurance}, we exploit this observation to relax \Cref{ax:smoothness} by comparing the loss in average quality directly with these gains.

\section{Applications}
\label{sec:applications}

This section applies and generalizes the main results of \Cref{sec:main} to four settings: lane pricing, waiting in line, public provision of goods with heterogeneous quality, and insurance under adverse selection.

\subsection{Lane Pricing}
\label{sec:lane} 

We begin with lane pricing, our leading application and primary motivation for the framework.  Recall that there are $N\geq2$ lanes; if a mass $m_i$ of agents is assigned to lane $i\in\{1,2,\ldots,N\}$, each agent in that lane experiences travel time $w(m_i)$.  Congestion effects are captured by $w:[0,1]\to[\und t,\BAR t]$, which is continuously differentiable, increasing, and convex, with $0\leq\und t<\BAR t<1$.  A travel time $t$ corresponds to effective quality $v(1-t)$, where $v$ is continuously differentiable, increasing, and concave. 

The laissez-faire allocation in this setting provides a natural candidate for the equity benchmark.  In the absence of pricing, agents distribute themselves evenly across lanes, resulting in the lane assignment $\bm_0 = \(1/N,1/N,\ldots,1/N\)$. The corresponding laissez-faire distribution assigns full probability to the quality $Q_0\coloneq v(1-w(1/N))$. 

In this setting, the laissez-faire allocation uniquely maximizes average quality across all feasible allocations, as shown in \Cref{sec:framework}. The convexity and strict monotonicity of the congestion technology induce a mean--dispersion tradeoff: any dispersion in qualities requires a departure from equal lane usage, which necessarily decreases average quality.  At the same time, this tradeoff is locally mild.  Specifically, the decrease in average quality induced by small amounts of dispersion is of higher order, as formally demonstrated in \Cref{app:lanes}.

Given that \Cref{ax:randomization,ax:friction,ax:smoothness} hold in this setting, our main results apply: \Cref{thm:two,thm:three} imply that there exists a strict Pareto improvement with three tiers, but no Pareto improvement with two tiers.  As shown in \Cref{app:lanes}, our main results extend to deterministic lane assignments: 

\begin{proposition}\label{prop:lane} 
When $N\geq3$, there exists a mechanism using only deterministic lane assignments that strictly Pareto-improves on the laissez-faire outcome.  When $N=2$, no mechanism using only deterministic lane assignments can be a Pareto improvement.
\end{proposition}

\Cref{prop:lane} has a simple practical interpretation. With at least three physical lanes, the Pareto improvement can be implemented deterministically: commuters can be assigned to fast, regular, and slow lanes, with no need to randomize their access. Prices can then be calibrated so that the regular lane has the same level of congestion as in the laissez-faire allocation.  

With only two physical lanes, some randomization is necessary to generate a Pareto improvement, but this need not require literal lotteries. It could instead be implemented through restrictions on access over time---for example, by allowing commuters choosing the intermediate tier to use the fast lane only on certain days or for a limited amount of time.

Our positive results rely on the ability to redistribute toll revenues to commuters. As emphasized in the congestion-pricing literature (e.g., \citealp{hall18}), such redistribution may be difficult in practice. Importantly, however, our proof of \Cref{thm:three} implies that redistribution need not be perfect. Fixing the distribution $G$ of commuters' types, there exists a cutoff $\BAR \alpha<1$ such that a Pareto improvement remains possible whenever a fraction $\alpha>\BAR \alpha$ of toll revenues can be returned to commuters. Such redistribution could be implemented indirectly, for example through reductions in other taxes or fees paid by drivers, or through credits deposited into electronic toll accounts that can ultimately be redeemed for cash.\footnote{What our framework abstracts from is the possible extensive-margin responses to the introduction of priority pricing. By definition, a strict Pareto improvement among existing commuters increases the attractiveness of driving and may therefore induce additional traffic.}

\subsection{Waiting in Line}
\label{sec:line}

We next consider a designer who can \emph{endogenously create lines} by partitioning agents and dedicating processing capacity to each line, as in airport security screening or access to public services. Unlike our lane-pricing application, where congestion depends only on the mass assigned to each lane, here congestion depends on demand relative to capacity in each line.

In this setting, agents arrive at a constant rate $\la>0$, and the designer has total processing capacity $\mu>0$.  Denote the demand--capacity ratio by $\rho\coloneq\la/\mu\in(0,1)$, which we refer to as the \emph{load} of the entire system.  A deterministic line design consists of a finite number of lines $K\in\{1,2,\ldots\}$ together with routing shares $\bm=(m_1,\ldots,m_K)\in\Delta^{K-1}$ and capacity shares $\bs=(s_1,\ldots,s_K)\in\Delta^{K-1}$.  Given such a design $(K,\bm,\bs)$, line $i\in\{1,2,\ldots,K\}$ is assigned arrival rate $\la m_i$ and capacity $\mu s_i$, so its load is 
\[\rho_i \coloneq \frac{\la m_i}{\mu s_i} = \rho\frac{m_i}{s_i}.\]
We restrict attention to designs in which $m_i,s_i>0$ and $\rho_i<1$ for every $i$.
 
We model delay using a reduced-form congestion technology.  Each agent assigned to line $i$ faces the deterministic waiting time $w(\rho_i)$, where $w:(0,1)\to[\und t,\BAR t]$ is continuously differentiable, increasing, and convex, with $0\leq\und t<\BAR t<1$.  The key property of this technology is \emph{scale-freeness}: if both demand and capacity in a line are scaled by the same factor, the load $\rho_i$ is unchanged and so is the wait time.  As in the lane-pricing application, the quality associated with wait time $w(\rho_i)$ is $v(1-w(\rho_i))$.

Next, we specify the feasible set of expected quality distributions.  We identify each deterministic line design $(K,\bm,\bs)$ with the quality distribution that it induces,
\[F_{(K,\bm,\bs)} \coloneq \sum_{i=1}^Km_i \d_{v(1-w(\rho_i))}.\] 
The designer can randomize agents' assignments within a deterministic line design, so the feasible set is 
\[\calF = \left\{ F\in\Delta([0,1]): F\in\MPC(F_{(K,\bm,\bs)}) \text{ for some deterministic line design }(K,\bm,\bs) \right\}.\]

Under the laissez-faire outcome, all agents use a single line.  This corresponds to a design with $K=1$ and $m_1=s_1=1$.  Each agent receives quality $Q_0=v(1-w(\rho))$, so the induced distribution is $\d_{Q_0}$.

All three axioms hold in this setting, with the laissez-faire expected quality distribution as the equity benchmark.  \Cref{ax:randomization} holds by construction; \Cref{app:line} verifies \Cref{ax:friction,ax:smoothness} by an argument similar to the lane-pricing application.  Intuitively, making one line faster requires allocating more capacity to it relative to demand, leaving another line with less capacity relative to demand.  Because waiting time is convex in line load, this dispersion increases average wait time.  Moreover, small departures from proportional capacity allocation create first-order dispersion at a higher-order cost to the average.

Since \Cref{ax:randomization,ax:friction,ax:smoothness} hold in this setting, \Cref{thm:two,thm:three} apply.  As in the lane-pricing application, the results extend to deterministic assignments to lines:
\begin{proposition}\label{prop:line} 
There exists a mechanism using only deterministic assignments to three lines that strictly Pareto-improves on the laissez-faire outcome.  No mechanism using only deterministic assignments to at most two lines can be a Pareto improvement. 
\end{proposition}

The implementation suggested by \Cref{prop:line} is again simple.  The designer can partition existing capacity across three lines, assigning relatively more capacity per user to the fast line, relatively less to the slow line, and preserving the laissez-faire demand--capacity ratio in the intermediate line.  In applications, this can be implemented through separate queues or service windows with different capacity allocations.

\subsection{Public Provision With Heterogeneous Quality} 
\label{sec:provision}

Our third application considers environments in which the designer allocates goods of heterogeneous physical quality. It illustrates how the mean--dispersion tradeoff underlying our results can arise from economic frictions unrelated to congestion.

We consider a unit mass of agents who demand a good of heterogeneous physical quality, $q\in[0,1]$.  For example, $q$ might represent the size of a public housing unit.  Each agent privately observes his type $r\in[\und r,\BAR r]\subset\R_+$.  An agent of type $r$ who receives physical quality $q$ and pays price $p$ obtains utility $rv(q)-p$, where $v:[0,1]\to[0,1]$ is continuously differentiable, increasing, and strictly concave.

A capacity-constrained designer has $q_0\in(0,1)$ units of physical quality to allocate per capita.  For instance, $q_0$ might represent the total amount of space that a public housing authority can subdivide into individual units.  Given an allocation function $q:[\und r,\BAR r]\to[0,1]$, feasibility requires that the total physical quality allocated not exceed available capacity: 
\[\int_{\und r}^{\BAR r}q(r)\,\dd G(r) \leq q_0.\] 

To apply the framework of \Cref{sec:framework}, we express this constraint in terms of effective quality.  Because $Q=v(q)$ and $v$ is increasing, providing effective quality $Q$ requires physical quality $v^{-1}(Q)$.  Accordingly, the feasible set of expected effective quality distributions is 
\[\calF = \left\{F\in\Delta([0,v(1)]): \E_F\!\left[v^{-1}(Q)\right]\leq q_0\right\}.\]
Because $v$ is concave, any lottery delivering expected effective quality $Q$ uses at least $v^{-1}(Q)$ units of expected physical quality; assigning physical quality $v^{-1}(Q)$ with certainty attains this bound.

In the absence of pricing, the designer cannot use quality to screen agents; a natural candidate for the equity benchmark is therefore uniform provision at physical quality $q_0$. This allocation provides every agent with effective quality $Q_0\coloneq v(q_0)$.

The feasible set satisfies all three axioms.  \Cref{ax:randomization} follows because $v^{-1}$ is convex, so the capacity constraint is preserved under mean-preserving contractions.  We verify \Cref{ax:friction} below and \Cref{ax:smoothness} in \Cref{app:provision}.  For any feasible distribution $F$, Jensen's inequality gives 
\[\E_F[Q] = \E_F\!\left[v\!\left(v^{-1}(Q)\right)\right] \leq v\!\(\E_F\!\left[v^{-1}(Q)\right]\) \leq v(q_0) = Q_0.\] 
Because $v$ is strictly concave, equality holds if and only if physical quality is almost surely constant at $q_0$.  Thus, the unique maximizing feasible distribution is $\d_{Q_0}$.  Any dispersion in physical quality around the equity benchmark necessarily decreases average \textit{effective} quality due to decreasing marginal utility from physical quality, yielding the key mean--dispersion tradeoff formalized by \Cref{ax:friction}. 

Since \Cref{ax:randomization,ax:friction,ax:smoothness} hold in this public-provision environment, our main results apply: 
\begin{proposition}\label{prop:quality} 
Any Pareto improvement over the equity benchmark must offer at least three levels of expected quality.  Moreover, there exists a strict Pareto improvement that offers three levels of expected quality. 
\end{proposition} 

Public housing provides a natural interpretation of this result.  Public housing programs often restrict households to a relatively limited set of options, making uniform provision a natural benchmark (e.g., \citealp{olsen03}; \citealp{sitaramanalstott19}).  Singapore's public housing system provides a useful contrast: the public housing authority chooses among different types of apartments to build and allocate, allowing households to express preferences over these alternatives (e.g., \citealp{ferdowsianetal26}). Our results suggest that differentiation---in particular, in size---need not come at the expense of equity. If prices are appropriately adjusted, offering at least three quality levels can generate a Pareto improvement over uniform provision, even if dispersion in physical quality lowers average effective quality (e.g., due to decreasing marginal utility over size). 

\subsection{Insurance Under Adverse Selection} 
\label{sec:insurance} 

Our final application shows how our main results generalize beyond the axioms introduced in \Cref{sec:framework}.  We consider an insurance program with a fixed budget, where the main friction generating the mean--dispersion tradeoff arises from adverse selection. Employer-provided health insurance (e.g., \citealp{einavetal10}) provides a natural example: a self-insured employer with a fixed benefits budget may offer plans of different generosity to employees whose willingness to pay for coverage is positively correlated with their expected medical costs.

Unlike the preceding applications, the resource cost of coverage depends on the type of the agent who receives it due to adverse selection.  We first argue that the feasible set satisfies \Cref{ax:randomization,ax:friction}.  We then show that \Cref{ax:smoothness} does not hold.  Finally, we provide a more general condition under which a three-tier Pareto improvement nevertheless exists. 

Let $Q\in[0,1]$ denote the generosity of insurance coverage.  A type-$r$ agent who receives coverage $Q$ and makes a premium adjustment $p$ obtains utility $rQ-p$.  Providing one unit of coverage to this agent generates expected cost $c(r)$ to the designer, where $c:[\und r,\BAR r]\to\R_{++}$ is continuous and increasing.  Thus, willingness to pay and expected claims are positively associated, as in a standard adverse-selection environment. A simple example is when $r$ reflects the probability of using the insurance, while $c(r)$ also incorporates the expected expenditure conditional on use.

The equity benchmark provides uniform coverage $Q_0\in(0,1)$ with no premium adjustments.  An allocation $Q:[\und r,\BAR r]\to[0,1]$ is feasible if the aggregate expected cost does not exceed that under the equity benchmark: 
\[\int_{\und r}^{\BAR r}c(r)Q(r)\,\dd G(r) \leq Q_0 \int_{\und r}^{\BAR r}c(r)\,\dd G(r).\]
Premium adjustments must satisfy budget balance.  To express the feasibility constraint in terms of distributions of expected quality, define $C(u)\coloneq c(G^{-1}(u))$ for every $u\in[0,1]$.  Because incentive compatibility assigns higher coverage to higher types, an expected quality distribution $F$ results in
\[\int_{\und r}^{\BAR r} c(r)Q(r)\,\dd G(r) = \int_0^1C(u)F^{-1}(u)\,\dd u\] 
units of expected cost.  Accordingly, the feasible set is 
\[\calF = \left\{F\in\Delta([0,1]): \int_0^1C(u)F^{-1}(u)\,\dd u \leq Q_0\int_0^1C(u)\,\dd u \right\}.\] 

We first verify \Cref{ax:randomization}.  For any $\bar F\in\MPC(F)$, since $\dd C$ is a nonnegative measure under our maintained assumptions, integration by parts yields
\begin{align*}
\int_0^1 C(u)\bar F^{-1}(u)\,\dd u 
&= C(1)\int_0^1 \bar F^{-1}(z)\,\dd z - \int_0^1\int_0^u \bar F^{-1}(z)\,\dd z\,\dd C(u)\\
&\leq C(1)\int_0^1 F^{-1}(z)\,\dd z - \int_0^1\int_0^u F^{-1}(z)\,\dd z\,\dd C(u) =\int_0^1 C(u) F^{-1}(u)\,\dd u.
\end{align*}
Thus, every mean-preserving contraction of a feasible distribution remains feasible, and \Cref{ax:randomization} holds. 

Next, we verify \Cref{ax:friction}.  Since both $C$ and $F^{-1}$ are nondecreasing, Chebyshev's integral inequality gives 
\[\int_0^1 C(u)F^{-1}(u)\,\dd u \geq \left[\int_0^1C(u)\,\dd u\right] \left[\int_0^1F^{-1}(u)\,\dd u\right] = \E_F[Q]\int_0^1C(u)\,\dd u.\] Feasibility therefore implies that $\E_F[Q]\leq Q_0$.  Because $C$ is increasing, this inequality is strict whenever $F$ is nondegenerate.  Hence, the distribution $\d_{Q_0}$ uniquely maximizes average quality, and \Cref{ax:friction} holds. 

Unlike in the preceding applications, however, the mean--dispersion tradeoff is not locally mild, due to adverse selection. Increasing coverage for high types is more expensive than reducing the same amount of coverage for low types saves, so feasibility requires a non-negligible additional reduction in coverage at the bottom. As a result, the loss in average quality cannot be of lower order than the quality gains at the top. In \Cref{app:insurance}, we show that \Cref{ax:smoothness} fails.

Even without \Cref{ax:smoothness}, however, a three-tier strict Pareto improvement remains possible if high types value coverage sufficiently highly relative to their expected cost.  We consider the condition 
\begin{equation}\label{eq:adverse-B}
\frac{\BAR r}{c(\BAR r)} > \frac{\und r}{c(\und r)}. \tag{B} 
\end{equation} 
The ``bang-for-the-buck'' ratio $r/c(r)$ measures a type's willingness to pay for coverage per unit of expected cost.  While this ratio may vary nonmonotonically with type, condition~\eqref{eq:adverse-B} requires this ratio to be higher at the upper endpoint than at the lower endpoint.

\begin{proposition}
\label{prop:insurance} 
There does not exist a two-tier Pareto improvement over uniform coverage.  However, under condition~\eqref{eq:adverse-B}, there exists a three-tier strict Pareto improvement.
\end{proposition} 

In the context of employer-provided health insurance, \Cref{prop:insurance} implies that---with a fixed benefits budget and budget-balanced premium adjustments---offering a choice between a low-coverage and a high-coverage plan cannot make all employees better off compared to a simple uniform-coverage scheme. A third, intermediate-coverage plan can achieve a Pareto improvement when employees with the highest willingness to pay for coverage have sufficiently high willingness to pay relative to their expected claims.

Without some restriction on the relationship between willingness to pay and expected cost, a Pareto improvement need not exist.  Suppose, for example, that $\und r>0$ and $c(r)=r$.  For any feasible, budget-balanced mechanism,
\[\int_{\und r}^{\BAR r}\left[rQ(r)-p(r)-rQ_0\right]\,\dd G(r) = \int_{\und r}^{\BAR r}r\left[Q(r)-Q_0\right]\,\dd G(r) = \int_{\und r}^{\BAR r}c(r)\left[Q(r)-Q_0\right]\,\dd G(r)\leq 0.\]
A Pareto improvement would make the first integrand nonnegative for every type and strictly positive for a positive mass of types, contradicting this inequality.  Thus, no Pareto improvement exists in this case, regardless of the number of tiers.

Condition~\eqref{eq:adverse-B} motivates a generalization of \Cref{ax:smoothness}.  While our framework in \Cref{sec:framework} deliberately abstracted from variable production costs for simplicity, such costs can be accommodated:

\setcounter{reaxiomno}{2}
\begin{reaxiom}[Generalized Smoothness]\label{ax:g-smoothness}
There exists a continuous function $\kappa:[0,1]\to\R_{++}$ such that:
\begin{enumerate}[label=\emph{(\roman*)}]
\item \label{it:g-1} For every bounded nondecreasing function $h:[0,1]\to\R$ satisfying $\int_0^1 \kappa(u)h(u)\,\dd u<0$, there exists a sequence $\{(\e_n,F_n)\}_{n=1}^\infty$ such that $\e_n\downarrow 0$, $F_n\in\calF$, and
\[\int_0^1\left\lvert\frac{F_n^{-1}(u)-Q_0}{\e_n} - h(u)\right\rvert\,\dd u\to0.\]
\item \label{it:g-2} The extreme types satisfy
\[\frac{G^{-1}(1)}{\kappa(1)} > \frac{G^{-1}(0)}{\kappa(0)}.\]
\end{enumerate}
\end{reaxiom}

\setcounter{rethmno}{1}
\begin{retheorem}\label{thm:g-three}
Under Axioms~\ref{ax:randomization} and \ref{ax:g-smoothness}\textsuperscript{*}, there exists a strict Pareto improvement with three tiers.
\end{retheorem}
\begin{proof}
See \Cref{app:additional}.
\end{proof}

Intuitively, \Cref{ax:g-smoothness}\textsuperscript{*} replaces the requirement that the loss in average quality be of higher order with a condition on the local resource tradeoff.  The function $\kappa(u)$ can be interpreted as the marginal resource cost of increasing expected quality at quantile $u$.  Condition~\emph{\ref{it:g-1}} requires every bounded, nondecreasing perturbation that strictly reduces resource use to be approximated by arbitrarily small feasible departures from the equity benchmark.  Condition~\emph{\ref{it:g-2}} requires the highest type to have a higher willingness to pay per unit of marginal resource cost than the lowest type.  Together, these conditions make it possible to reduce quality for a small mass of low types and increase quality for a small mass of high types, while retaining $Q_0$ as the intermediate tier.  The resulting gains can then be redistributed through payments so that every type is strictly better off.

In our insurance application, condition~\emph{\ref{it:g-1}} of \Cref{ax:g-smoothness}\textsuperscript{*} holds with $\kappa=C$, while condition~\emph{\ref{it:g-2}} is equivalent to condition~\eqref{eq:adverse-B}.  This shows that \Cref{ax:g-smoothness}\textsuperscript{*} can hold even when \Cref{ax:smoothness} fails.  In contrast, under \Cref{ax:randomization,ax:friction}, \Cref{ax:smoothness} implies \Cref{ax:g-smoothness}\textsuperscript{*} (see  \Cref{app:additional}).

\section{Concluding Remarks}
\label{sec:conclusion}

This paper shows that, under economically natural assumptions, three priority tiers are necessary and sufficient for a Pareto improvement over the equity benchmark.  With only two tiers, the agent indifferent between them receives no sorting gain to offset the loss in average quality and is therefore worse off.  A third tier separates the threshold between the low and intermediate tiers from the threshold between the intermediate and high tiers.  This separation allows sorting gains to outweigh the loss in average quality for every agent.

Our framework isolates this logic without specifying interpersonal welfare comparisons or committing to a particular allocation technology.  It applies to lane pricing, waiting-line design, and public provision, and extends to insurance under adverse selection, where the resource cost of quality depends on the recipient.  More broadly, our results show that the distributional consequences of priority pricing depend not only on whether prices are introduced, but also on how the priority tiers are designed.  Policy debates framed as a choice between uniform access and a two-tier priority system therefore overlook a simple alternative: a third tier can allow the gains from priority pricing to be shared by all agents.

\bibliography{master_bibliography.bib}
\bibliographystyle{econ-econometrica}

\clearpage
\appendix

\section{Omitted Proofs}
\label{app:proofs}

\subsection{Proofs for \texorpdfstring{\Cref{sec:lane}}{Section 4.1}}\label{app:lanes}

We first argue that \Cref{ax:smoothness} holds in this application. Consider a perturbation of the laissez-faire allocation.  For $\e\in(0,1/N)$, define 
\[\bm_\e \coloneq \(\frac1N-\e,\frac1N,\ldots,\frac1N,\frac1N+\e\),\] 
so that a mass $\e$ of agents is shifted from lane 1 to lane $N$, while all other lanes retain their laissez-faire masses.  Writing $F^\e$ for the induced distribution of qualities, the resulting loss in average quality is 
\[Q_0-\E_{F^\e}[Q] = 2\phi\!\(\frac1N\)-\phi\!\(\frac1N-\e\)-\phi\!\(\frac1N+\e\)=\calo(\e).\] 
The final equality follows from the differentiability of $\phi$: the two first-order effects cancel.  By contrast, the increase in dispersion---measured by the cumulative improvement above $Q_0$ received by agents in lane 1---is 
\[\E_{F^\e}\!\left[\(Q-Q_0\)_+\right] = \(\frac1N-\e\)\left[v\!\(1-w\!\(\frac1N-\e\)\)-v\!\(1-w\!\(\frac1N\)\)\right] = \frac{\e}Nv'\!\(1-w\!\(\frac1N\)\)w'\!\(\frac1N\)+\calo(\e).\] 
Because $v',w'>0$, this expression is positive for all sufficiently small $\e$ and is of order $\e$.  Hence, \Cref{ax:smoothness} holds:
\[\lim_{\e\downarrow0} \frac{Q_0-\E_{F^\e}[Q]}{\E_{F^\e}[\(Q-Q_0\)_+]} = 0.\] 

We now prove \Cref{prop:lane}. Suppose first that $N=2$.  Any deterministic lane assignment induces a quality distribution with at most two tiers.  A one-tier distribution cannot be a Pareto improvement: incentive compatibility requires all payments to be equal, budget balance requires them to be zero, and feasibility implies that the common quality is no greater than $Q_0$.  A two-tier distribution cannot be a Pareto improvement by \Cref{thm:two}, since \Cref{ax:friction} holds in this setting.  Hence, no mechanism using only deterministic lane assignments can be a Pareto improvement when $N=2$.

Now suppose that $N\geq3$.  For $\e\in(0,1/N)$, consider the deterministic lane assignment
\[\bm_\e\coloneq\(\frac1N-\e,\frac1N,\ldots,\frac1N,\frac1N+\e\).\]
The induced distribution $F^\e$ has three tiers.  Writing
\[Q_L\coloneq v\!\(1-w\!\(\frac1N+\e\)\)\AND Q_H\coloneq v\!\(1-w\!\(\frac1N-\e\)\),\]
these tiers are $Q_L<Q_0<Q_H$, with respective masses $1/N+\e$, $(N-2)/N$, and $1/N-\e$.

Consider the mechanism implementing $F^\e$.  The threshold types between the low and intermediate tiers and between the intermediate and high tiers are
\[r_L=G^{-1}\!\(\frac1N+\e\)\AND r_H=G^{-1}\!\(1-\frac1N+\e\).\]
By \cref{eq:middle-tier-gain}, the intermediate-tier payment $p_0$ satisfies
\[-p_0=\(\frac1N-\e\)r_H\(Q_H-Q_0\)-\(\frac1N+\e\)r_L\(Q_0-Q_L\).\]
As $\e\downarrow0$,
\[\begin{dcases}
Q_H-Q_0&=v'\!\(1-w\!\(\frac1N\)\)w'\!\(\frac1N\)\e+\calo(\e),\\
Q_0-Q_L&=v'\!\(1-w\!\(\frac1N\)\)w'\!\(\frac1N\)\e+\calo(\e).
\end{dcases}\]
Therefore,
\[-p_0 = \frac{\e}{N}v'\!\(1-w\!\(\frac1N\)\)w'\!\(\frac1N\)\left[G^{-1}\!\(1-\frac1N\)-G^{-1}\!\(\frac1N\)\right]+\calo(\e).\]
Because $N\geq3$, we have $1-1/N>1/N$.  Since $G^{-1}$ is strictly increasing and $v',w'>0$, it follows that $-p_0>0$ for all sufficiently small $\e>0$.

By the argument in \Cref{sec:proofs}, the positive intermediate-tier subsidy implies that the mechanism strictly Pareto-improves on the laissez-faire outcome for all sufficiently small $\e>0$.

\subsection{Proofs for \texorpdfstring{\Cref{sec:line}}{Section 4.2}}\label{app:line}

We first prove that \Cref{ax:friction,ax:smoothness} hold. As in our lane-pricing application, since randomization does not affect average effective quality, it suffices to show that every deterministic line design has average effective quality at most $Q_0$ and that every design attaining $Q_0$ induces the distribution $\d_{Q_0}$.  To this end, define $\phi(y)\coloneq yv(1-w(\rho y))$; since $w$ is increasing and convex on $(0,1)$, it follows that $\phi:(0,1/\rho)\to\R$ is strictly concave.  For any deterministic line design $(K,\bm,\bs)$, average effective quality can be written as 
\[\sum_{i=1}^Ks_i \phi\!\(\frac{m_i}{s_i}\)\leq \phi\!\(\sum_{i=1}^Ks_i\frac{m_i}{s_i}\) = \phi(1) = v(1-w(\rho)) = Q_0.\] 
Since $\phi$ is strictly concave, equality holds if and only if $m_i/s_i$ is constant across lines.  Because $\sum_{i=1}^Km_i=\sum_{i=1}^Ks_i=1$, this constant must be 1.  Hence, equality holds if and only if $m_i=s_i$ for every $i$, in which case every line has load $\rho$ and the induced distribution is $\d_{Q_0}$.  Thus, \Cref{ax:friction} holds.

\Cref{ax:smoothness} is also satisfied in this setting.  To see this, fix $\eta\in(0,1/2)$.  For any $\e>0$ sufficiently small that $\rho\(1+\e\)<1$, consider a deterministic three-line design with capacity and routing shares 
\[(s^F,s^0,s^S) = (\eta,1-2\eta,\eta) \AND (m^F,m^0,m^S) = \(\eta\(1-\e\),1-2\eta,\eta\(1+\e\)\).\] 
The corresponding loads are $\rho\(1-\e\)$, $\rho$, and $\rho\(1+\e\)$, so the induced distribution $F_\e$ has three support points 
\[v\!\(1-w\!\(\rho\(1-\e\)\)\) > Q_0 = v(1-w(\rho)) > v\!\(1-w\!\(\rho\(1+\e\)\)\).\] 
These have corresponding masses $\eta\(1-\e\)$, $1-2\eta$, and $\eta\(1+\e\)$.  By construction, the loss in average effective quality relative to the equity benchmark is 
\[Q_0-\E_{F_\e}[Q] = \eta\left[2\phi(1)-\phi(1-\e)-\phi(1+\e)\right] = \calo(\e).\] 
By contrast, the cumulative improvement above $Q_0$ received by agents in the fast line is 
\[\E_{F_\e}\!\left[\(Q-Q_0\)_+\right] = \eta\(1-\e\) \left[v\!\(1-w\!\(\rho\(1-\e\)\)\)-v(1-w(\rho))\right]= \eta\rho v'(1-w(\rho))w'(\rho)\e+\calo(\e). \] 
Because $v',w'>0$, this expression is positive for all sufficiently small $\e$ and is of order $\e$.  Consequently, \Cref{ax:smoothness} holds:
\[\lim_{\e\downarrow0} \frac{Q_0-\E_{F_\e}[Q]}{\E_{F_\e}[\(Q-Q_0\)_+]}=0.\] 

We now sketch the proof of \Cref{prop:line}, which closely parallels the proof of \Cref{prop:lane}.

No mechanism using only deterministic assignments to at most two lines can be a Pareto improvement.  Such assignments induce at most two levels of expected quality.  A one-tier distribution cannot be a Pareto improvement by the argument in \Cref{sec:main}; a two-tier distribution cannot be a Pareto improvement by \Cref{thm:two}, since \Cref{ax:friction} holds in this setting.

Now suppose $K=3$.  Fix $\eta\in(0,1/2)$ and, for $\e>0$ sufficiently small that $\rho(1+\e)<1$, consider the deterministic three-line design
\[(s^F,s^0,s^S)=(\eta,1-2\eta,\eta) \AND (m^F,m^0,m^S)=\(\eta(1-\e),1-2\eta,\eta(1+\e)\),\]
already used above to verify \Cref{ax:smoothness}.  As shown there, the induced distribution $F_\e$ assigns mass $\pi_H\coloneq\eta(1-\e)$ to $Q_H\coloneq v(1-w(\rho(1-\e)))$, mass $1-2\eta$ to $Q_0$, and mass $\pi_L\coloneq\eta(1+\e)$ to $Q_L\coloneq v(1-w(\rho(1+\e)))$, with $Q_L<Q_0<Q_H$.

Consider the mechanism implementing $F_\e$.  The threshold types between the slow and intermediate lines and between the intermediate and fast lines are
\[r_L\coloneq G^{-1}\!\(\eta(1+\e)\)\AND r_H\coloneq G^{-1}\!\(1-\eta(1-\e)\).\]
By \cref{eq:middle-tier-gain}, the intermediate-tier payment $p_0$ satisfies
\[-p_0=\pi_Hr_H\(Q_H-Q_0\)-\pi_Lr_L\(Q_0-Q_L\).\]
By the same argument as used above,
\[Q_H-Q_0=\rho\,v'\!\(1-w(\rho)\)w'(\rho)\,\e+\calo(\e)\AND Q_0-Q_L=\rho\,v'\!\(1-w(\rho)\)w'(\rho)\,\e+\calo(\e).\]
As $\e\downarrow0$, $r_L\to G^{-1}(\eta)$ and $r_H\to G^{-1}(1-\eta)$, so
\[-p_0=\eta\rho\,v'\!\(1-w(\rho)\)w'(\rho)\left[G^{-1}(1-\eta)-G^{-1}(\eta)\right]\e+\calo(\e).\]
Because $\eta\in(0,1/2)$, we have $1-\eta>\eta$; since $G^{-1}$ is strictly increasing and $v',w'>0$, it follows that $-p_0>0$ for all sufficiently small $\e>0$.

By the argument in \Cref{sec:proofs}, the positive intermediate-tier subsidy implies that the mechanism strictly Pareto-improves on the laissez-faire outcome for all sufficiently small $\e>0$.

\subsection{Proofs for \texorpdfstring{\Cref{sec:provision}}{Section 4.3}}\label{app:provision}

We show that \Cref{ax:smoothness} holds.  Fix $m\in(0,1/2)$ and, for $\e>0$ sufficiently small, consider a three-tier allocation that assigns mass $m$ to physical quality $q_0+\e$, mass $m$ to physical quality $q_0-\e$, and the remaining mass to physical quality $q_0$.  This perturbation preserves average physical quality and is therefore feasible.  Writing the induced distribution of effective qualities as 
\[F_\e = m\d_{v(q_0-\e)} + \(1-2m\)\d_{v(q_0)} + m\d_{v(q_0+\e)},\] 
the corresponding loss in average effective quality is 
\[Q_0-\E_{F_\e}[Q] = m\left[2v(q_0)-v(q_0-\e)-v(q_0+\e)\right] = \calo(\e).\] 
However, the cumulative improvement above $Q_0$ received by the mass $m$ of agents assigned physical quality $q_0+\e$ is 
\[\E_{F_\e}\!\left[\(Q-Q_0\)_+\right] = m\left[v(q_0+\e)-v(q_0)\right] = mv'(q_0)\e + \calo(\e).\] 
Consequently, \Cref{ax:smoothness} holds, and \Cref{prop:quality} follows from \Cref{thm:two,thm:three}:
\[\lim_{\e\downarrow0} \frac{Q_0-\E_{F_\e}[Q]}{\E_{F_\e}[\(Q-Q_0\)_+]} = 0.\] 

\subsection{Proofs for \texorpdfstring{\Cref{sec:insurance}}{Section 4.4}}\label{app:insurance}

We first show that the setting of \Cref{sec:insurance} does not satisfy \Cref{ax:smoothness}.
For any feasible $F$ such that $\E_F[\(Q-Q_0\)_+]>0$, feasibility implies that both $\Pr_F[Q<Q_0]$ and $\Pr_F[Q>Q_0]$ are positive.  Chebyshev's integral inequality then implies that
\[\begin{dcases}
\int_{1-\Pr_F[Q>Q_0]}^1C(u)\left[F^{-1}(u)-Q_0\right]\,\dd u &\geq \frac{\E_F[\(Q-Q_0\)_+]}{\Pr_F[Q>Q_0]}\int_{1-\Pr_F[Q>Q_0]}^1C(u)\,\dd u,\\
\int_0^{\Pr_F[Q<Q_0]}C(u)\left[Q_0-F^{-1}(u)\right]\,\dd u &\leq \frac{\E_F[\(Q_0-Q\)_+]}{\Pr_F[Q<Q_0]}\int_0^{\Pr_F[Q<Q_0]}C(u)\,\dd u.
\end{dcases}\]
Feasibility requires 
\begin{align*}
&\int_{1-\Pr_F[Q>Q_0]}^1C(u)\left[F^{-1}(u)-Q_0\right]\,\dd u \leq \int_0^{\Pr_F[Q<Q_0]}C(u)\left[Q_0-F^{-1}(u)\right]\,\dd u\\
&\implies \frac{\E_F[\(Q_0-Q\)_+]}{\Pr_F[Q<Q_0]}\int_0^{\Pr_F[Q<Q_0]}C(u)\,\dd u\geq \frac{\E_F[\(Q-Q_0\)_+]}{\Pr_F[Q>Q_0]}\int_{1-\Pr_F[Q>Q_0]}^1C(u)\,\dd u.
\end{align*}
Consequently, 
\[\frac{Q_0-\E_F[Q]}{\E_F[\(Q-Q_0\)_+]} = \frac{\E_F[\(Q_0-Q\)_+]}{\E_F[\(Q-Q_0\)_+]} - 1 \geq \frac{\Pr_F[Q<Q_0]\int_{1-\Pr_F[Q>Q_0]}^1C(u)\,\dd u}{\Pr_F[Q>Q_0]\int_0^{\Pr_F[Q<Q_0]}C(u)\,\dd u} - 1.\]
The expression on the right is the ratio of the average of $C$ over the uppermost $\Pr_F[Q>Q_0]$ quantiles to its average over the lowermost $\Pr_F[Q<Q_0]$ quantiles, minus one.  Moreover, this lower bound is uniform.  Viewed as a function of the two probabilities, the expression on the right extends continuously to the compact set of pairs in $[0,1]^2$ whose sum is at most one, where the lower and upper averages at zero equal $C(0)$ and $C(1)$, respectively.  It is strictly positive throughout this set because $C$ is strictly increasing, and so it has a strictly positive minimum.  Thus, \Cref{ax:smoothness} does not hold in this setting.

We now prove \Cref{prop:insurance}.

Because \Cref{ax:friction} holds, \Cref{thm:two} immediately implies that there does not exist a two-tier Pareto improvement over uniform coverage.

We next establish the positive result. Under condition~\eqref{eq:adverse-B},
\[\lim_{m\downarrow0}\frac{G^{-1}(1-m)}{\frac1m\int_{1-m}^1C(u)\,\dd u}=\frac{\BAR r}{c(\BAR r)}>\frac{\und r}{c(\und r)}=\lim_{m\downarrow0}\frac{G^{-1}(m)}{\frac1m\int_0^mC(u)\,\dd u}.\]
Hence, by continuity, there exists $m\in(0,1/2)$ such that
\[\frac{G^{-1}(1-m)}{\frac1m\int_{1-m}^1C(u)\,\dd u}>\frac{G^{-1}(m)}{\frac1m\int_0^mC(u)\,\dd u}.\]
Define
\[Q_L\coloneq Q_0-\frac{\e}{m}\int_{1-m}^1C(u)\,\dd u\AND Q_H\coloneq Q_0+\frac{\e}{m}\int_0^mC(u)\,\dd u.\]
For $\e>0$ sufficiently small, consider the three-tier distribution
\[F^\e\coloneq m\d_{Q_L}+\(1-2m\)\d_{Q_0}+m\d_{Q_H}.\]
Because $Q_0\in(0,1)$, all three support points lie in $[0,1]$ and are distinct when $\e$ is sufficiently small.

This distribution is feasible. Indeed, relative to uniform coverage, its aggregate expected cost is
\begin{align*}
\int_0^1C(u)\left[(F^\e)^{-1}(u)-Q_0\right]\dd u
&=-\frac{\e}{m}\left[\int_0^mC(u)\,\dd u\right]\left[\int_{1-m}^1C(u)\,\dd u\right]\\
&\qquad+\frac{\e}{m}\left[\int_{1-m}^1C(u)\,\dd u\right]\left[\int_0^mC(u)\,\dd u\right]=0.
\end{align*}

The implementing mechanism strictly Pareto-improves on uniform coverage.  Let $r_L=G^{-1}(m)$ and $r_H=G^{-1}(1-m)$ denote the two threshold types.  By \cref{eq:middle-tier-gain} and our earlier choice of $m$, the intermediate-tier premium adjustment $p_0$ satisfies
\begin{align*}
-p_0
&=mr_H\(Q_H-Q_0\)-mr_L\(Q_0-Q_L\)=\e\left[r_H\int_0^m C(u)\,\dd u-r_L\int_{1-m}^1 C(u)\,\dd u\right]>0.
\end{align*}
By the argument in \Cref{sec:proofs}, the positive intermediate-tier subsidy implies that $F^\e$ yields a three-tier strict Pareto improvement over uniform coverage.

\section{Additional Results and Generalizations}
\label{app:additional}

\subsection{Proof of \texorpdfstring{\Cref{thm:g-three}\textsuperscript{*}}{Theorem 2*}}

By condition~\emph{\ref{it:g-2}} and continuity,
\[\lim_{m\downarrow0}\frac{G^{-1}(1-m)}{\frac1m\int_{1-m}^1\kappa(u)\,\dd u}=\frac{G^{-1}(1)}{\kappa(1)}>\frac{G^{-1}(0)}{\kappa(0)}=\lim_{m\downarrow0}\frac{G^{-1}(m)}{\frac1m\int_0^m\kappa(u)\,\dd u}.\]
Hence, there exists $m\in(0,1/2)$ such that
\[\frac{G^{-1}(m)}{G^{-1}(1-m)}<\frac{\int_0^m\kappa(u)\,\dd u}{\int_{1-m}^1\kappa(u)\,\dd u}.\]
Fix such an $m$ and define the bounded nondecreasing function
\[h(u)\coloneq
\begin{dcases}
\frac12\left[\frac{G^{-1}(m)}{G^{-1}(1-m)}+\frac{\int_0^m\kappa(z)\,\dd z}{\int_{1-m}^1\kappa(z)\,\dd z}\right],& 1-m\leq u\leq 1,\\
0,& m\leq u<1-m,\\
-1,& 0\leq u<m.
\end{dcases}\]
By construction,
\[\int_0^1\kappa(u)h(u)\,\dd u=-\int_0^m\kappa(u)\,\dd u+h(1)\int_{1-m}^1\kappa(u)\,\dd u<0.\]
By condition~\emph{\ref{it:g-1}}, we obtain a sequence $\{(\e_n,F_n)\}_{n=1}^\infty$ such that $\e_n\downarrow0$, $F_n\in\calF$, and
\begin{equation}\label{eq:bound}
\int_0^1\left\lvert\frac{F_n^{-1}(u)-Q_0}{\e_n}-h(u)\right\rvert\,\dd u\to0.
\end{equation}

Next, for each $n$, we construct a three-tier mean-preserving contraction $\hat F_n^{(3)}$ of $F_n$ by pooling qualities within three quantile intervals.  To this end, define
\[Q_{L,n}\coloneq \frac1m\int_0^mF_n^{-1}(u)\,\dd u,\qquad Q_{M,n}\coloneq\frac1{1-2m}\int_m^{1-m}F_n^{-1}(u)\,\dd u,\qquad Q_{H,n}\coloneq\frac1m\int_{1-m}^1F_n^{-1}(u)\,\dd u.\]
Consider
\[\hat F_n^{(3)}\coloneq m\d_{Q_{L,n}}+\(1-2m\)\d_{Q_{M,n}}+m\d_{Q_{H,n}}.\]
Since $\hat F_n^{(3)}$ preserves the conditional means within each quantile interval, $\hat F_n^{(3)}\in\MPC(F_n)$ and is feasible by \Cref{ax:randomization}.  Moreover, \cref{eq:bound} implies that
\[Q_{L,n}=Q_0-\e_n+\calo(\e_n),\qquad Q_{M,n}=Q_0+\calo(\e_n),\qquad Q_{H,n}=Q_0+h(1)\e_n+\calo(\e_n).\]
Thus, for all sufficiently large $n$, 
\[Q_{L,n}<Q_0<Q_{H,n}\AND Q_{L,n}<Q_{M,n}<Q_{H,n}.\]

We now modify $\hat F_n^{(3)}$ to obtain a three-tier distribution $F_n^{(3)}$ whose intermediate tier is exactly $Q_0$.  If $Q_{M,n}=Q_0$, no further pooling is required.  We thus consider two cases:
\begin{enumerate}
\item If $Q_{M,n}<Q_0$, pool the entire mass at $Q_{M,n}$ with enough of the mass at $Q_{H,n}$ for the pooled mass to have mean $Q_0$. 
\item If $Q_{M,n}>Q_0$, pool the entire mass at $Q_{M,n}$ with enough of the mass at $Q_{L,n}$ for the pooled mass to have mean $Q_0$.    
\end{enumerate}
The masses taken from the respective tails are
\[\frac{\(1-2m\)\abs{Q_{M,n}-Q_0}}{Q_{H,n}-Q_0}\AND \frac{\(1-2m\)\abs{Q_{M,n}-Q_0}}{Q_0-Q_{L,n}}.\]
Since $Q_{M,n}=Q_0+\calo(\e_n)$, these masses converge to zero for all sufficiently large $n$.  Consequently, this pooling is possible for all sufficiently large $n$.  We write the resulting distribution as 
\[F_n^{(3)} = \pi_{L,n}\d_{Q_{L,n}} + \pi_{0,n}\d_{Q_0} + \pi_{H,n}\d_{Q_{H,n}},\qquad\text{where }\pi_{L,n}+\pi_{0,n} + \pi_{H,n}=1.\]
Because $F_n^{(3)}$ is a mean-preserving contraction of $\hat F_n^{(3)}$, it is feasible by \Cref{ax:randomization}.  Moreover, for all sufficiently large $n$, it has exactly three tiers with strictly positive masses.

Now, consider the mechanism implementing $F_n^{(3)}$.  Let $r_{L,n}\coloneq G^{-1}(\pi_{L,n})$ and $r_{H,n}\coloneq G^{-1}(1-\pi_{H,n})$ denote the two threshold types.  By construction,
\[r_{L,n}\to G^{-1}(m)\AND r_{H,n}\to G^{-1}(1-m).\]
The implementing mechanism yields a strict Pareto improvement for all sufficiently large $n$.  By \cref{eq:middle-tier-gain}, the intermediate-tier payment $p_{0,n}$ satisfies
\[-p_{0,n}=\pi_{H,n}r_{H,n}\(Q_{H,n}-Q_0\)-\pi_{L,n}r_{L,n}\(Q_0-Q_{L,n}\).\]
It follows that
\[\frac{-p_{0,n}}{\e_n}\to m\left[G^{-1}(1-m)h(1)-G^{-1}(m)\right]>0.\]
Hence, the intermediate tier is subsidized for all sufficiently large $n$.  By the argument in \Cref{sec:proofs}, $F_n^{(3)}$ therefore yields a three-tier strict Pareto improvement.

\subsection{Proof of \texorpdfstring{\Cref{lem:axiom}}{Lemma 1}}

In this appendix, we prove that \Cref{ax:g-smoothness}\textsuperscript{*} is weaker than \Cref{ax:smoothness}.

\begin{lemma}\label{lem:axiom}
Suppose \Cref{ax:randomization,ax:friction,ax:smoothness} hold.  Then \Cref{ax:g-smoothness}\textsuperscript{*} is satisfied by choosing $\kappa\equiv 1$.
\end{lemma}

Suppose that \Cref{ax:randomization,ax:friction,ax:smoothness} hold.  Set $\kappa\equiv1$; this function is continuous and strictly positive.  Under this choice of $\kappa$, condition~\emph{\ref{it:g-2}} holds:
\[\frac{G^{-1}(1)}{\kappa(1)}=\BAR r>\und r=\frac{G^{-1}(0)}{\kappa(0)}.\]
To verify condition~\emph{\ref{it:g-1}}, fix any bounded nondecreasing function $h:[0,1]\to\R$ such that
\[\int_0^1h(u)\,\dd u<0.\]
By \Cref{ax:smoothness}, for every $n$, there exists $\hat F_n\in\calF$ such that
\[\E_{\hat F_n}[\(Q-Q_0\)_+]>0\AND\frac{Q_0-\E_{\hat F_n}[Q]}{\E_{\hat F_n}[\(Q-Q_0\)_+]}\leq\frac1n.\]
Clearly, $\hat F_n\ne\d_{Q_0}$; thus, \Cref{ax:friction} implies that $Q_0-\E_{\hat F_n}[Q]>0$.  Since qualities lie in $[0,1]$,
\[0<Q_0-\E_{\hat F_n}[Q]\leq\frac1n\E_{\hat F_n}\!\left[(Q-Q_0)_+\right]\leq\frac1n.\]
Define
\begin{equation}\label{eq:eps}
\e_n\coloneq\frac{Q_0-\E_{\hat F_n}[Q]}{-\int_0^1h(u)\,\dd u}>0.
\end{equation}
Passing to a subsequence and relabeling if necessary, we may take $\e_n\downarrow0$.

Next, observe that $Q_0\in(0,1)$.  Indeed, if $Q_0=1$, then $\E_F[\(Q-Q_0\)_+]=0$ for every $F\in\calF$, which contradicts \Cref{ax:smoothness}.  In addition, if $Q_0=0$, any distribution satisfying $\E_F[\(Q-Q_0\)_+]>0$ would have strictly positive mean---and hence higher average quality than the equity benchmark, which contradicts the definition of $Q_0$.

Because $h$ is bounded and nondecreasing, for all sufficiently large $n$ there exists $F_n\in\Delta([0,1])$ whose quantile function satisfies $F_n^{-1}(u)=Q_0+\e_nh(u)$, almost everywhere.
By construction, $F_n$ shares the same mean with $\hat F_n$:
\[\E_{F_n}[Q]=Q_0+\e_n\int_0^1h(u)\,\dd u=\E_{\hat F_n}[Q].\]

We now show that $F_n$ is a mean-preserving contraction of $\hat F_n$ for all sufficiently large $n$.  To see this, fix any point $t$ between the lowest and highest support points of $F_n$.  Then
\[\E_{F_n}[\(Q-t\)_+]\leq 2\e_n\norm{h}_\infty \AND\E_{\hat F_n}[\(Q-t\)_+]\geq\E_{\hat F_n}[\(Q-Q_0\)_+]-\e_n\norm{h}_\infty.\]
Moreover, \cref{eq:eps} implies that
\[\frac{\e_n}{\E_{\hat F_n}[\(Q-Q_0\)_+]}=\frac{Q_0-\E_{\hat F_n}[Q]}{-\int_0^1h(u)\,\dd u\cdot\E_{\hat F_n}[\(Q-Q_0\)_+]}\leq\frac1{-n\int_0^1h(u)\,\dd u}\to0.
\]
Hence, for all sufficiently large $n$ and for all $t$ lying between the lowest and highest support points of $F_n$,
\[\E_{F_n}[\(Q-t\)_+]\leq\E_{\hat F_n}[\(Q-t\)_+].\]
If $t$ lies below the support of $F_n$, then equality of means gives
\[\E_{F_n}[\(Q-t\)_+]=\E_{F_n}[Q]-t=\E_{\hat F_n}[Q]-t\leq\E_{\hat F_n}[\(Q-t\)_+].\]
If $t$ lies above the support of $F_n$, then the same inequality holds trivially:
\[\E_{F_n}[\(Q-t\)_+]=0\leq \E_{\hat F_n}[\(Q-t\)_+].\] 
Thus, this inequality holds for every $t\in\R$; hence, $F_n\in\MPC(\hat F_n)$.  Since $\hat F_n$ is feasible, \Cref{ax:randomization} implies that $F_n$ is feasible as well.

Finally, observe that for all sufficiently large $n$, our definition of $F_n$ satisfies
\[\int_0^1\left\lvert\frac{F_n^{-1}(u)-Q_0}{\e_n}-h(u)\right\rvert\,\dd u=0.\]
Passing to a subsequence and relabeling if necessary, we obtain the sequence $\{(\e_n,F_n)\}_{n=1}^\infty$ required by condition~\emph{\ref{it:g-1}}.  Hence, \Cref{ax:g-smoothness}\textsuperscript{*} holds with $\kappa\equiv1$. 
\end{document}

%% file: working_arxiv.tex
\usepackage{amssymb,amsthm,mathtools,bbm}
\usepackage{enumitem,cite}
\usepackage[longnamesfirst]{natbib}
\usepackage{multirow,abstract}
\usepackage{framed,mdframed,cancel,setspace}
\usepackage[bottom,hang]{footmisc}
\usepackage{graphicx,tikz}
\usepackage[colorlinks=true,allcolors=NavyBlue,citecolor=NavyBlue]{hyperref}
\hypersetup{breaklinks=true}
\usetikzlibrary{calc,positioning}
\usepackage[font={small,it}]{caption}
\usepackage[noabbrev]{cleveref}
\usepackage[T1]{fontenc}

\usepackage{breakurl}

\newtheoremstyle{ans}{15pt}{20pt}{}{}{\bfseries}{}{ }{\thmname{#1}\thmnote{ #3}.}
\theoremstyle{ans}
\newtheorem*{soln}{Solution}

\newcounter{lemno}
\newcounter{app_lemno}
\newcounter{propno}
\newcounter{app_propno}
\newcounter{thmno}
\newcounter{clmno}
\newcounter{app_thmno}
\newcounter{defno}
\newcounter{factno}
\newcounter{corno}
\newcounter{assno}
\newcounter{reassno}
\newcounter{examplenum}
\newcounter{remarkno}
\newcounter{axiomno}
\newcounter{reaxiomno}
\newcounter{rethmno}
\newmdenv[backgroundcolor=blue!7,linewidth=0pt]{bBox}
\newmdenv[backgroundcolor=green!7,linewidth=0pt]{gBox}
\newmdenv[backgroundcolor=red!7,linewidth=0pt]{rBox}
\newmdenv[backgroundcolor=gray!7,linewidth=0pt]{grayBox}

\newtheoremstyle{statementStyle}% % Theorem style name
{12pt}% Space above
{12pt}% Space below
{}% % Body font
{}% Indent amount
{}% % Theorem head font
{{\bf .}\;}% Punctuation after theorem head
{0.25em}% Space after theorem head
{{\bf\thmname{#1}\thmnumber{ #2}}% Theorem text (e.g. Theorem 2.1)
\thmnote{ (#3)}} % Optional theorem note
\makeatother

\theoremstyle{statementStyle}

\theoremstyle{plain}
\newtheorem{lemma}[lemno]{Lemma}

\newtheorem{theorem}[thmno]{Theorem}
\newtheorem{axiom}[axiomno]{Axiom}

\newtheorem{corollary}[corno]{Corollary}
\newtheorem{exerciseT}[thmno]{Exercise}

\newtheorem{definition}[defno]{Definition}

\Crefname{assumption}{Assumption}{Assumptions}
\Crefname{axiom}{Axiom}{Axioms}
\newtheorem{proposition}[propno]{Proposition}
\Crefname{proposition}{Proposition}{Propositions}

\Crefname{remark}{Remark}{Remarks}
\newtheorem*{example*}{Example}

\Crefname{claim}{Claim}{Claims}

\newtheoremstyle{nbstatementStyle}% % Theorem style name
{12pt}% Space above
{12pt}% Space below
{\it}% % Body font
{}% Indent amount
{}% % Theorem head font
{{\bf .}\;}% Punctuation after theorem head
{0.25em}% Space after theorem head
{{\bf\thmname{#1}}% Theorem text (e.g. Theorem 2.1)
\thmnote{ (#3)}} % Optional theorem note
\makeatother

\theoremstyle{nbstatementStyle}

\newtheoremstyle{restatementStyle}
  {\topsep}   % ABOVESPACE
  {\topsep}   % BELOWSPACE
  {\itshape}  % BODYFONT
  {0pt}       % INDENT (empty value is the same as 0pt)
  {\bfseries} % HEADFONT
  {.}         % HEADPUNCT
  {5pt plus 1pt minus 1pt} % HEADSPACE
  {\thmname{#1}\thmnumber{ #2\textsuperscript{*}}% Theorem text (e.g. Theorem 2.1)
  \thmnote{ (#3)}} % Optional theorem note
  
\theoremstyle{restatementStyle}

\newtheorem{reaxiom}[reaxiomno]{Axiom}
\newtheorem{retheorem}[rethmno]{Theorem}

\Crefname{reaxiom}{Axiom}{Axioms}
\Crefname{retheorem}{Theorem}{Theorems}

\newtheorem{prob}{Problem}

\newtheorem{conj}{Conjecture}
\newtheorem{prope}{Property}

\newtheoremstyle{reexerciseStyle}% % Theorem style name
{12pt}% Space above
{12pt}% Space below
{}% % Body font
{}% Indent amount
{\bfseries}% % Theorem head font
{{\bf .}\;}% Punctuation after theorem head
{0.25em}% Space after theorem head
{\thmname{#1}\thmnumber{ #2}% Theorem text (e.g. Theorem 2.1)
\thmnote{ (#3)}} % Optional theorem note
\makeatother

\newtheoremstyle{exerciseStyle}% % Theorem style name
{12pt}% Space above
{12pt}% Space below
{}% % Body font
{}% Indent amount
{\bfseries}% % Theorem head font
{{\bf .}\;}% Punctuation after theorem head
{0.25em}% Space after theorem head
{\thmname{#1}\thmnumber{ #2}% Theorem text (e.g. Theorem 2.1)
\thmnote{ (#3)}} % Optional theorem note
\makeatother

\theoremstyle{exerciseStyle}
\newtheoremstyle{teach}{15pt}{\topsep}{}{}{\itshape}{}{ }{\thmname{#1}\thmnote{ #3}.}
\theoremstyle{teach}
\newtheorem*{prooflemma}{Proof of lemma}

\newtheoremstyle{rem}{5pt}{0pt}{\color{black}}{}{\bfseries}{}{ }{\thmname{#1}\thmnote{ #3}.}
\theoremstyle{rem}

\definecolor{darkgreen}{rgb}{0.0, 0.75, 0.2}
\makeatletter
\renewcommand*\env@matrix[1][*\c@MaxMatrixCols c]{%
  \hskip -\arraycolsep
  \let\@ifnextchar\new@ifnextchar
  \array{#1}}
\makeatother
\renewcommand{\d}{\delta}

\newcommand{\dd}{\mathrm{d}}

\newcommand{\la}{\lambda}

\newcommand{\e}{\varepsilon}

\newcommand{\R}{\mathbf{R}}

\newcommand{\x}{\mathbf{x}}

\newcommand{\BAR}{\overline}

\newcommand{\und}[1]{\underline{#1}}

\DeclarePairedDelimiter{\norm}{\lVert}{\rVert} 
\DeclarePairedDelimiter{\abs}{\lvert}{\rvert}

\newcommand{\AND}{\quad\text{and}\quad}

\renewcommand{\(}{\left(}
\renewcommand{\)}{\right)}

\newcommand{\bs}{{\boldsymbol{\s}}}

\DeclareMathOperator{\supp}{supp}

\DeclareMathOperator{\E}{\mathbf{E}}

\def\calF{\mathcal{F}}

\usepackage{datetime}

\newdateformat{monthyeardate}{%
  \monthname[\THEMONTH] \THEYEAR}
\newdateformat{daymonthyeardate}{%
  \monthname[\THEMONTH] \THEDAY, \THEYEAR}

\title{\LARGE\hspace*{\fill}\\[1ex]\papertitle\\[2ex]}
\author{\large\name\\[1ex] \affiliation}
\date{\monthyeardate\today}

\emergencystretch=\maxdimen
\makeatletter
\def\@xfootnote[#1]{%
  \protected@xdef\@thefnmark{#1}%
  \@footnotemark\@footnotetext}
\makeatother

\makeatletter
\ifFN@para
\else
  \long\def\@makefntext#1{%
    \ifFN@hangfoot
      \bgroup
      \setbox\@tempboxa\hbox{%
        \ifdim\footnotemargin>0pt
          \hb@xt@\footnotemargin{\@makefnmark\hss}%
        \else
          \@makefnmark\hskip-\footnotemargin      %%Changed here
        \fi
      }%
      \leftmargin\wd\@tempboxa
      \rightmargin\z@
      \linewidth \columnwidth
      \advance \linewidth -\leftmargin
      \parshape \@ne \leftmargin \linewidth
      \footnotesize
      \@setpar{{\@@par}}%
      \leavevmode
      \llap{\box\@tempboxa}%
      \parskip\hangfootparskip\relax
      \parindent\hangfootparindent\relax
    \else
      \parindent1em
      \noindent
      \ifdim\footnotemargin>\z@
        \hb@xt@ \footnotemargin{\hss\@makefnmark}%
      \else
        \ifdim\footnotemargin=\z@
          \llap{\@makefnmark}%
        \else
          \llap{\hb@xt@ -\footnotemargin{\@makefnmark\hss}}%
        \fi
      \fi
    \fi
    \footnotelayout#1%
    \ifFN@hangfoot
      \par\egroup
    \fi
  }
\fi
\makeatother

\makeatletter
\def\@endtheorem{\endtrivlist}% NEW
\makeatother